\documentclass[12pt]{amsart}

\usepackage{graphicx}
\usepackage{amsmath,amssymb}

\def\>{\relax\ifmmode\mskip.666667\thinmuskip\relax\else\kern.111111em\fi}
\def\<{\relax\ifmmode\mskip-.333333\thinmuskip\relax\else\kern-.0555556em\fi}
\def\vsk#1>{\vskip#1\baselineskip}
\def\vv#1>{\vadjust{\vsk#1>}\ignorespaces}
\def\vvn#1>{\vadjust{\nobreak\vsk#1>\nobreak}\ignorespaces}

\newtheorem{thm}{Theorem}[section]
\newtheorem{cor}[thm]{Corollary}
\newtheorem{lem}[thm]{Lemma}

\theoremstyle{definition}                                  
\numberwithin{equation}{section}

\theoremstyle{definition}
\newtheorem*{rem}{Remark}

\numberwithin{equation}{section}

\usepackage{amssymb}
\usepackage{comment}
\usepackage[all]{xy}
\usepackage{enumerate}

\usepackage[all]{xy}
\usepackage{bbm}

\usepackage[dvipsnames]{xcolor}
\usepackage{graphicx}

\usepackage[colorlinks=true,allcolors = blue]{hyperref} 

\usepackage{amsmath}
\usepackage{amsxtra}
\usepackage{amscd}
\usepackage{amsthm}
\usepackage{amsfonts}
\usepackage{amssymb}
\usepackage{eucal}
\usepackage{verbatim}
\usepackage[all]{xy}
\usepackage{mathrsfs}
\usepackage{lineno}
\usepackage{color}
\definecolor{deepjunglegreen}{rgb}{0.0, 0.29, 0.29}
\definecolor{darkspringgreen}{rgb}{0.09, 0.45, 0.27}

\usepackage{stmaryrd}

\usepackage{etex,etoolbox,needspace}
\pretocmd\section{\Needspace*{4\baselineskip}}{}{}

\DeclareMathOperator{\Char}{{char}}

\usepackage{tikz-cd}

\newcommand{\rar}[1]{\stackrel{#1}{\longrightarrow}}

\newcommand{\al}{\alpha}  \newcommand{\ga}{\gamma}

\newcommand{\om}{\omega} \newcommand{\Om}{\Omega}

\newcommand{\bA}{{\mathbb A}}

\newcommand{\bC}{{\mathbb C}}

\newcommand{\bF}{{\mathbb F}}
\newcommand{\bG}{{\mathbb G}}

\newcommand{\bK}{{\mathbb K}}

\newcommand{\bP}{{\mathbb P}}
\newcommand{\bQ}{{\mathbb Q}}

\newcommand{\bZ}{{\mathbb Z}}

\newcommand{\cF}{{\mathcal F}}

\newcommand{\cH}{{\mathcal H}}

\newcommand{\cL}{{\mathcal L}}

\newcommand{\cO}{{\mathcal O}}
\newcommand{\cP}{{\mathcal P}}

\newcommand{\cU}{{\mathcal U}}
\newcommand{\cV}{{\mathcal V}}

\newcommand{\fO}{{\mathfrak O}}

\newcommand{\Yb}{\overline{Y}}

\newcommand{\nc}{\newcommand}

\nc\wh{\widehat}

\nc\on{\operatorname}

\nc\Gr{\on{Gr}}

\nc\Fl{\on{Fl}}

\DeclareMathOperator{\Lie}{{Lie}}

\DeclareMathOperator{\Mat}{{Mat}}

\DeclareMathOperator{\coker}{{coker}}

\newcommand{\limfrom}{{\displaystyle\lim_{\longleftarrow}}}
\newcommand{\limto}{{\displaystyle\lim_{\longrightarrow}}}
\newcommand{\rightlim}{\mathop{\limto}}

\newcommand{\leftlim}{\mathop{\displaystyle\lim_{\longleftarrow}}}
\newcommand{\limfromn}{\leftlim\limits_{\raise3pt\hbox{$n$}}}
\newcommand{\limton}{\rightlim\limits_{\raise3pt\hbox{$n$}}}

\newcommand{\rightlimit}[1]{\mathop{\lim\limits_{\longrightarrow}}\limits%
                    _{\raise3pt\hbox{$\scriptstyle #1$}}}

\newcommand{\leftlimit}[1]{\mathop{\lim\limits_{\longleftarrow}}\limits%
                    _{\raise3pt\hbox{$\scriptstyle #1$}}}

\newcommand{\iso}{\buildrel{{}_\sim}\over{\longrightarrow}}
\newcommand{\mono}{\hookrightarrow}

\DeclareMathOperator{\Id}{{Id}}

\DeclareMathOperator{\Hom}{{Hom}} 

\DeclareMathOperator{\Spec}{{Spec}}

\DeclareMathOperator{\Res}{{Res}}

\makeatletter

\newcommand{\Rmnum}[1]{\expandafter\@slowromancap\romannumeral #1@}
\makeatother

\newtheorem{pr}{Proposition}[section]

\theoremstyle{definition}

\newcommand{\Div}{\operatorname{div}}

\def\vsk#1>{\vskip#1\baselineskip}
\def\<{\relax\ifmmode\mskip-.333333\thinmuskip\relax\else\kern-.0555556em\fi}

\nc{\bea}{\begin{eqnarray*}}
\nc{\eea}{\end{eqnarray*}}
\nc{\bean}{\begin{eqnarray}}
\nc{\eean}{\end{eqnarray}}

\let\bi\bibitem
\def\kz{K\!Z}
\let\der\partial
\let\mc\mathcal
\def\F{{\mathbb F}}
\def\C{{\mathbb C}}
\def\ox{{\otimes}}
\def\Z{\mathbb Z}
\def\E{\mathbb E}

\renewcommand{\thefootnote}{[\arabic{footnote}]}

\def\K{\mathbb K}
\def\L{\mathbb L}

\begin{document}

\title[KZ equations in characteristic $p$]
{Finding all solutions to  the  KZ  equations
\\
 in characteristic $p$}
\author[Alexander Varchenko]{Alexander Varchenko$^{\star}$}
\author[Vadim Vologodsky]{Vadim Vologodsky$^\circ$}

\maketitle

\begin{center}

{\it $^{\star}$Department of Mathematics, University
of North Carolina at Chapel Hill\\ Chapel Hill, NC 27599-3250, USA\/}

\vsk.5>

{\it ${}^\circ$ Department of Mathematics, Princeton University\\
Princeton, NJ 08540, USA\/}

\end{center}

{\let\thefootnote\relax
\footnotetext{\vsk-.8>\noindent
{$^\star$\sl E-mail}:\enspace anv@email.unc.edu,
supported in part by NSF grant  DMS - 1954266
\\
{$^\circ$\sl E-mail}:\enspace vologod@gmail.com, supported by a Simons Foundation Investigator Grant No.~622511 through Bhargav Bhatt. 
}}

\begin{abstract}

The KZ equations are differential equations satisfied by the correlation functions (on the Riemann sphere) of two-dimensional conformal field theories associated with an affine Lie algebra at a fixed level. They form a system of complex partial differential equations with regular singular points satisfied by the $n$-point functions of affine primary fields.
In \cite{SV1} the KZ equations were identified with equations for flat sections of suitable Gauss-Manin connections, and solutions of the KZ equations were constructed in the form of multidimensional hypergeometric integrals.
In \cite{SV2} the KZ equations were considered modulo a prime number $p$, and, for rational levels, polynomial solutions of the KZ equations modulo $p$ were constructed by an elementary procedure as suitable $p$-approximations of the hypergeometric integrals. 
In this paper we study in detail
 the first nontrivial example of the KZ equations in characteristic $p$. In particular, if the level is irrational,  we prove a version of the steepest descent result that relates the KZ local system to the
space of functions on the critical locus of the master function. We use this result to prove the generic irreducibility of the KZ local system at any irrational level.
If the level is rational, we describe all solutions of the KZ equations in characteristic $p$ by demonstrating that they all stem from the $p$-hypergeometric solutions.
Finally, we prove a Lagrangian property of the subbundle of the KZ bundle spanned by the $p$-hypergeometric sections.

\end{abstract}

{\small \tableofcontents  }

\section{Introduction}
\label{sec 1}

The Knizhnik–Zamolodchikov equations, or KZ equations, are linear differential equations satisfied by the correlation functions (on the Riemann sphere) of two-dimensional conformal field theories associated with an affine Lie algebra at a fixed level. 
They form a system of complex partial differential equations with regular singular points satisfied by the $n$-point functions of affine primary fields.
The equations were derived by physicists Knizhnik and Zamolodchikov in their seminal work \cite{KZ}.
In \cite{SV1} the KZ equations were identified with equations for flat sections of suitable Gauss-Manin connections, and solutions of the KZ equations were constructed in the form of multidimensional hypergeometric integrals.
In \cite{SV2} the KZ equations were considered modulo a prime number $p$, and polynomial solutions of the KZ equations modulo $p$ were constructed by an elementary procedure as suitable $p$-approximations of the hypergeometric integrals.
In this paper we address the problem of whether all solutions of the KZ equations modulo $p$ are generated by the $p$-hypergeometric 
 solutions. 
We consider the first nontrivial example of the KZ equations and 
demonstrate that, indeed, in this case, all solutions of the KZ equations modulo $p$ stem from the $p$-hypergeometric solutions.

\subsection{KZ equations}
\label{sec 1.1}
The KZ equations considered in this paper have the following form.
Let $\Om_{ij}$ be the $n\times n$-matrix with only four nonzero elements given by the formula
\small
\begin{equation}
\label{Omij}
 \Omega_{ij} \ = \ \begin{pmatrix}
             & \vdots^i &  & \vdots^j &  \\
        {\scriptstyle i} \cdots & {-1} & \cdots &
            1  & \cdots \\
                   & \vdots &  & \vdots &   \\
        {\scriptstyle j} \cdots & 1 & \cdots & {-1}&
                 \cdots \\
                   & \vdots &  & \vdots &
                   \end{pmatrix} .
\end{equation}               
\normalsize
Let $z=(z_1,\dots,z_n)$. 
Define 
\bea
H_i(z) = \sum_{j \ne i}
   \frac{\Omega_{ij}}{z_i - z_j}  ,
\quad i = 1, \dots , n,
\eea
called the Gaudin Hamiltonians.
Consider the system of differential and algebraic equations for 
 column vectors  $I(z)=(I_1(z)$, \dots, $I_{n}(z))^\intercal$:
 \bean
\label{KZ}
\phantom{aaaaaa}
\der_i I \ + \
   h\,H_i(z) \,I \,=\,0,
\quad i = 1, \dots , n,
\quad
I_1(z)+\dots+I_{n}(z)=0,
\eean
where $\der_i$ denotes $\frac{\partial }{\partial z_i}$; \
$M^\intercal$ denotes the transpose matrix of a matrix $M$; 
$h$ is a parameter of the system.
In this paper this systems of equations is called the KZ equations
 and denoted $\kz(h)$.

\smallskip

Let $\mathbb A^n$ be affine space with coordinates $z_1,\dots,z_n$.
Let
$S$ be the complement to the union of all diagonal hyperplanes defined by equations
$z_i=z_j$, $i\ne j$.
The KZ equations define a flat KZ connection $\nabla^{\on{KZ}, h}$
on the trivial bundle $\pi : V\times S\to S$ whose fiber is
the vector space $V$ of $n$-vectors with the zero sum of coordinates. We denote the sheaf of sections of this bundle
by
$\cV$.

\begin{rem}

The system of equations \eqref{KZ} is the system of the original KZ equations 
associated with the Lie algebra $\mathfrak{sl}_2$ and the subspace of singular vectors of weight $n-2$
of the tensor power $(\C^2)^{\ox n}$  of the two-dimensional irreducible $\mathfrak{sl}_2$-module, up to a gauge
transformation, see this example in \cite[Section 1.1]{V2}.

\end{rem}

In this paper we study the KZ  connection $\nabla^{\on{KZ}, h}$ in characteristic $p>0$. In particular, we describe explicitly all flat sections of 
the vector bundle 
$\cV$ with connection $\nabla^{\on{KZ}, h}$
and compute the $p$-curvature of $\nabla^{\on{KZ}, h}$.

\subsection{Dual connections}

The vector space $V$ has a non-degenerate bilinear form
\bean
\label{Sha}
(x,y) = x_1y_1+\dots + x_ny_n\,,
\eean
the Shapovalov form. The Gaudin Hamiltonians are symmetric,
\bea
(H_i(z) x,y) = (x,H_i(z)y)
\eea
for all $i$,  $x,y$.
The KZ connections $\nabla^{\on{KZ}, h}$ and $\nabla^{\on{KZ}, -h}$ are dual, that is,
\bean
\label{do}
\der_i (x(z), y(z)) = (\nabla^{\on{KZ}, h}_i x(z), y(z))+ (x(z), \nabla^{\on{KZ}, -h}_i y(z))
\eean
for any $x(z), y(z)$. Equivalently, form \eqref{Sha} defines an isomorphism of bundles with flat connections 
\begin{equation}\label{eq:dualitySha}
     (\cV, \nabla^{\on{KZ}, -h})\iso (\cV, \nabla^{\on{KZ}, h})^*.
\end{equation}

In particular, if $x(z)$ is a flat section of the connection $\nabla^{\on{KZ}, h}$
and $y(z)$ is a flat section of the connection $\nabla^{\on{KZ}, -h}$, then
$\der_i(x(z), y(z))=0$.

\smallskip

Let $y(z)$ be a nonzero $V$-valued function.  For a fixed $z$ let $W_{y(z)}\subset V$ be the hyperplane  defined by the formula
\bea
W_{\!y(z)} =\{x\in V\mid (x, y(z))=0\}.
\eea
Let $\mc W_{y(z)}$ be the codimension one subbundle of the trivial bundle 
$\pi : \mathbb A^n\times V\to \mathbb A^n$ with fibers $W_{\!y(z)}$.

\smallskip

If $y(z)$ is a flat section of the KZ connection $\nabla^{\on{KZ}, -h}$, then the subbundle
$\mc W_{y(z)}$ is invariant with respect to the KZ connection $\nabla^{\on{KZ}, h}$,
namely, if $x(z)$ is a section of $\mc W_{y(z)}$, then $\nabla^{\on{KZ}, h}_ix(z)$ also is a section of 
$\mc W_{y(z)}$  for any $i$. 

\smallskip

 Indeed,
\bea
 \big(\nabla^{\on{KZ}, h}_ix(z), y(z)\big)=\big(\nabla^{\on{KZ}, h}_ix(z), y(z)\big)+ 
\big(x(z), \nabla^{\on{KZ}, - h}_iy(z)) = \der_i\big(x(z),y(z)\big) = 0.
\eea

\subsection{Solutions over $\C$}

 Assume that 
\bea
h\ =\ -\, \frac r q\,,
\eea
where $r, q$ are coprime positive integers, $1\leq r<q$. Define the { master function}
\bean
\label{mast f}
\Phi(x,z) = \prod_{a=1}^{n}(x-z_a)^{-r/q}
\eean
and  the ${n}$-vector  of hypergeometric  integrals
\bean
\label{Iga}
I^{(\ga)} (z)=(I_1(z),\dots,I_n(z))^\intercal,\qquad I_j=\int_{\ga(z)} \frac{\Phi(x,z)dx}{x-z_j},
\eean
 where  $\ga(z)$ is a Pochhammer contour in $\C-\{z_1,\dots,z_n\}$.

\begin{thm}
\label{thm1.1}

 The vector $I^{(\ga)}(z)$ satisfies  the KZ equations  $\kz(-r/q)$.
\end{thm}

Theorem \ref{thm1.1} is a classical statement. 
Much more general algebraic and differential equations satisfied by analogous multidimensional hypergeometric integrals were considered in \cite{SV1}.  Theorem \ref{thm1.1} is discussed as an example in  \cite[Section 1.1]{V2}.

\begin{thm} [{\cite[Formula (1.3)]{V1}}]
\label{thm dim}

If $nr/q$ is not a positive integer, then all local solutions of  the equations $\kz(-r/q)$ have this form.
In particular, the complex vector space of local solutions of the form \eqref{Iga} is $n-1$-dimensional. 

\end{thm}

\subsection{Solutions in finite characteristic}\label{ss:solutions_in_finite}

Let $\K$ be a field of finite characteristic $p$.
Consider the equations $\kz(h)$ over $\K$. 

The case $h=0$ is trivial. Assume that  $h\in \F_p\subset \K$, $h\ne 0$.
 Let $\tilde h$ be  the representative of $h$ in $\Z$ such that
\bean
\label{h p}
1\leq \tilde  h\leq p-1.
\eean
Define
\bean
\label{def P}
P(x,z)=\prod_{s=1}^n(x-z_s)\ \in\ \K[x,z].
\eean
The function $P(x,z)^{\tilde h}$ is called the master function.
Consider the vector of polynomials
\bean
\label{Q}
Q(x,z,h) =P(x,z)^{\tilde h} \Big(\frac 1{x-z_1}, \dots,\frac 1{x-z_n}\Big)^\intercal \,=\, \sum_i Q^{(i)}(z,h) \,x^i,
\eean
where 
\bean
\label{Q coor}
Q^{(i)}(z,h)= (Q^{(i)}_1(z,h), \dots, Q^{(i)}_n(z,h))^\intercal
\eean
are vectors of polynomials in $z$ with coefficients in $\F_p$\,.

\begin{thm}[{\cite[Theorem 1.2]{SV2}}] 
\label{thm Fp}

For any positive integer $\ell$, the vector of polynomials
\\
 $Q^{(\ell p-1)}(z,h)$ satisfies the  equations $\kz(h)$ over $\K$\,.

\end{thm}

The vector   $Q^{(\ell p-1)}(z,h)$ is zero if   $\ell \not\in \big\{1, \dots, \big[\frac{n{\tilde h}}p\big]\big\}$ for degree reasons.
The vectors  
\bean
\label{p hyp}
Q^{(\ell p-1)}(z,h), \qquad 
\ell = 1, \dots, \Big[\frac{n{\tilde h}}p\Big],
\eean
 are called the $p$-hypergeometric solutions of
the equations $\kz(h)$ over $\K$\,.

Notice that if $\frac{n{\tilde h}}p < 1$, then there are no $p$-hypergeometric solutions.

\smallskip
Denote  
$\K[z^p]=\K[z_1^p,\dots,z_{n}^p]$.
The set of all solutions of $\kz(h)$ over $\K$   is a module over the ring  $\K[z^p]$.

\begin{thm}[\cite{SlV}]
\label{thm inde} Let $p>n$. 
Then the  $p$-hypergeometric solutions in \eqref{p hyp}
are linear independent over the ring $\K[z^p]$,
that is, if \,
\bea
\sum_{\ell =1}^{\big[\frac{n{\tilde h}}p\big]}   c_\ell(z)Q^{(\ell p-1)}(z,h)=0
\eea
for some $c_\ell(z)\in\K[z^p]$, then  $c_\ell(z)=0$ for all $\ell$. 

\end{thm}

Consider the dual equations $\kz(-h)$ over $\K$.
 The previous constructions give us the following corollary.

\begin{cor}
\label{cor -h}
The vectors
\bean
\label{p -h hyp}
Q^{(\ell p-1)}(z,-h), \qquad \ell = 1, \dots, \Big[\frac{n(p-\tilde h)}p\Big],
\eean
satisfy the equations $\kz(-h)$. If $p>n$, then these solutions are linear independent over $\K[z^p]$. 
\end{cor}

The total number of $p$-hypergeometric solutions of the equations $\kz(h)$ and $\kz(-h)$ equals
\bean
\label{nu s}
\Big[\frac{n{\tilde h}}p\Big] + \Big[\frac{n(p-{\tilde h})}p\Big] = n-1 = \dim V.
\eean

\subsection{What is done in this paper?}
There are two cases:  $h\in \K \backslash \bF_p$ and  $h\in \bF_p$\,. 

\smallskip

If  $h\in \K \backslash \bF_p$, we prove that the KZ connection is irreducible and, in particular, the KZ 
equations have no formal solutions. 
This is a striking difference with the picture over $\bC$. 
  
If $h\in \bF_p$\,, we show that all solutions are $p$-hypergeometric.  
We shall discuss each case in more detail in the remaining part of the introduction.

\subsubsection{Case  $h\in \K \backslash \bF_p$.}

Let $a=(a_1, \dots ,a_n)$ be a $\K$-point of $S$. Recall that a 
formal solution of the KZ equations at $a$ over $\K$ is a vector $I
\in \K[[(z_1-a_1, \dots  , z_n-a_n)]]^{\oplus n}$ satisfying equations \eqref{KZ}. 
In Corollary \ref{cor:no_formal_solutions} we prove the following statement.

\begin{thm} Let $p$ be a prime number that does not divide $n$.
    Assume that $h\in \K \backslash \bF_p$\,. Then, for every $a\in S(\K)$, the equations 
$\kz(h)$ have no formal solutions at $a$ over $\K$.
\end{thm}

In fact, we prove a stronger result. Consider the vector space $V\otimes  \K(z)$ over the field $\K(z)=\K(z_1,\dots , z_n)$ of rational functions regular on $S$.   The KZ connection $\nabla^{\on{KZ}, h}$ defines a $\K$-linear action 
on $V\otimes  \K(z)$ of the algebra $D_S$ of differential operators on $S$:
 for each $i$, the operator $\frac{\partial}{\partial z_i}$ acts by $\nabla^{\on{KZ}, h}_i$. 
\begin{thm}\label{thm:intro_irr} Assume that $p>n$ and  $h\in \K \backslash \bF_p$. Then the $D_S$-module  $V\otimes  \K(z)$ is irreducible.
\end{thm} 
We refer the reader to Theorem \ref{thm:irreducibility} for a proof.

\subsubsection{Case  $h\in \bF_p$.}
In this case we prove the following result.

\begin{thm}\label{thm:main3_introduction}

 Let $p$ be a prime number 
 that does not divide $n$, and $h\in \bF_p\subset \K $.
Then the bundle  $(\cV, \nabla^{\on{KZ}, h})$
has a trivial flat subbundle $\mc U$ of rank $\Big[\frac{n\tilde h}p\Big]$ 
such that the quotient $ \cV/\mc U$ is also a trivial flat bundle. 
The $p$-hypergeometric solutions of the equations $\kz(h)$ form a flat basis for $\mc U$.
The $p$-hypergeometric solutions of the equations $\kz(-h)$, regarded  
via Shapovalov form \eqref{eq:dualitySha}  as flat sections  of the dual bundle $\cV^*$, 
annihilate $\mc U$ and constitute a flat basis for $\big(\cV/\mc U\big)^*$. Finally, 
if $\mc U\ne 0$, then every flat section of  $\cV$ belongs to $\mc U$.  
\end{thm}

We suggest two proofs of Theorem \ref{thm:main3_introduction}: see Theorems  \ref{th:main_result_for_rational_h} and  \ref{thm 7}.

Theorem \ref{thm:main3_introduction} implies the following corollaries.

\begin{cor} 
\label{cor lin ind}
Assume that $p$ does not divide $n$.
For $a\in S(\overline{\K})$,
the vectors $Q^{(\ell p-1)}(a, h) \in V\ox \K,$ 
 $\ell = 1, \dots, \big[\frac{n{\tilde h}}p\big]$,
 are linearly independent.

\end{cor}

\begin{cor} [Orthogonality]
\label{thm ort} 
Assume that $p$ does not divide $n$ and $h\in\F_p,$  $h\ne 0$. Then the  $p$-hypergeometric solutions of the equations $\kz(h)$ and $\kz(-h)$ are orthogonal
under the Shapovalov form. Namely,
for any $\ell \in \big\{1, \dots, \big[\frac{n\tilde h}p\big]\big\}$ and $
m \in \big\{1, \dots, \big[\frac{n(p-\tilde h)}p\big]\big\}$
we have
\bean
\label{ort}
\big(Q^{(\ell p-1)}(z,h), Q^{(m p-1)}(z,-h)\big) = 0.
\eean

\end{cor}

An elementary proof  of Corollary \ref{thm ort}, valid for all $p$ and $n$, is given in Appendix 
\ref{App 1}.

\begin{cor}
\label{thm comp}

Assume that $p$ does not divide $n$ and that $\big[\frac{n \tilde h}{p}\big]>0$.
Then the space of global solutions of the equations 
$\kz(h)$ over $\K$ is a free module over the algebra 
$\K [z_i^p, (z_i-z_j)^{-p}, 1\leq i< j \leq n]$ of rank $\big[\frac{n\tilde h}{p}\big]$.
The $p$-hypergeometric solutions in \eqref{p hyp} 
form a basis of this module.

\end{cor}

\begin{cor}
Assume that $p$ does not divide $n$,  $\big[\frac{n\tilde h}{p}\big]>0$, and $a=(a_1,\dots,a_n)\in S(\K)$. 
Then the space of formal solutions of equations $\kz(h)$ at $a$
over $\K$ is a free module over the algebra 
$\K[[(z_1-a_1)^p, \dots, (z_n-a_n)^p]]$ of rank $\big[\frac{n \tilde h}{p}\big]$.
The $p$-hypergeometric solutions in \eqref{p hyp} 
form a basis of this module.

\end{cor}

See Corollary \ref{cor form}.

\subsubsection{$p$-curvature}
In characteristic $p$, the operators
\bea
\Psi_k(z) = \left(\nabla^{\on{KZ}, h}_k\right)^p, \qquad k=1,\dots, n,
\eea
 commute with multiplication by functions.
 The linear operators $\Psi_k(z)$ are called the  $p$-curvature operators
of the KZ connection $\nabla^{\on{KZ},h}$.

It is known that  the intersection of kernels  $\cap_{k=1}^n\on{ker} \Psi_k(z)$ of the $p$-curvature operators
coincides with the space generated by  flat sections of $\nabla^{\on{KZ},h}$,  see the Cartier descent theorem \cite[Theorem 5.1]{katz70}.

\begin{thm}
\label{thm Pcu}

Assume that $p$ does not divide $n$ and  $h\in\F_p$. Then
\bean
\label{CC=0}
\Psi_k(z) \Psi_\ell(z) = 0
\eean  
for  $k, \ell \in\{1,\dots, n\}$.
\begin{enumerate}
\item[$(\on{i})$]

If $h=0$, then all operators $\Psi_k(z)$ are equal to zero.

\item[$(\on{ii})$] If $h\in \F_p\setminus \{0\}$ and   $\big[\frac{n\tilde h}p\big]=0$
or $\big[\frac{n\tilde h}p\big]=n-1$, then all operators $\Psi_k(z)$ are equal to zero.

\item[$(\on{iii})$] If $ 0< \big[\frac{n\tilde h}p\big]<n-1$, then every
$p$-curvature operator $\Psi_k(z)$ is of rank 1.
Its kernel is defined by the following  linear equation:
\bean
\label{ker p}
\phantom{aaaaaa}
\on{ker} \Psi_k(z) = \left\{ 
(v_1,\dots,v_n)^\intercal \in V \ \Big\vert \  \sum_{i=1}^n v_i\sum_{m =1}
^{\big[\frac{n(p-\tilde h)}p\big]} 
z_k^{{p}(m-1)} Q^{(m p-1)}_i(z, -h) = 0 
\right\} ,
\eean
and its image is generated by the vector
\bean
\label{IM p}
\sum_{\phantom{aa}\ell =1}
^{\big[\frac{n \tilde h}p\big]} 
z_k^{{p}(\ell-1)} Q^{(\ell p-1)}(z, h),
\eean
see formula \eqref{5.7}.

\end{enumerate}

 \end{thm}
Formula \eqref{CC=0} and parts (i), (ii) follow easily from Theorem \ref{thm:main3_introduction}; see Lemma \ref{lem CC} for details. 
Part (iii) is proved in Theorem \ref{thm pPk} and Corollary \ref{cor 55}, as a consequence of the Katz p-curvature theorem \cite{katz}.


\smallskip

Notice that  the linear equation in \eqref{ker p} is satisfied by all $p$-hypergeometric flat sections
$Q^{(mp-1)}(z,h)$, $m=1,\dots, \big[\frac{nh}p\big]$, by Orthogonality Corollary \ref{thm ort}.

\smallskip


\smallskip
Using the linear equation in  \eqref{ker p} and orthogonality relations \eqref{ort}, 
it is easy to check directly that the intersection $\cap_{k=1}^n\on{ker} \Psi_k(z)$ of kernels of the $p$-curvature operators 
 coincides with the span
of $p$-hypergeometric flat sections.

\subsection{Geometric interpretation of the KZ  connection}

 Our approach to the results stated above is based on the geometric interpretation
of the KZ  connection as the Gauss-Manin connection. To explain this interpretation, fix a ground field $\K$, $h\in \K$ and consider the polynomial $P=\prod_{i=1}^n (x-z_i)$ viewed as  a rational function on $\bP^1_S: =\bP^1 \times S$.
Define an integrable logarithmic connection  on the trivial bundle by the formula 
$$\cO_{\bP^1_S} \rar{\nabla}  \Omega^1_{\bP^1_S}(\log T), \quad  \nabla(f) = df + h f \frac{dP}{P},$$
where $T\mono \bP^1_S$ is the support of $\Div P$. We denote by $\cP^h$ the logarithmic local system $(\cO_{\bP^1_S}, \nabla)$ over $\bP^1_S$. Then $(\cV, \nabla^{\on{KZ}, -h})$ is identified with 
the relative  logarithmic de Rham cohomology $E_h:=H^1_{dR, \log}(\bP^1_S/S,\cP^{h})$, see Proposition \ref{pr:expliciteformulafortheKZconn}. This interpretation goes back to \cite{DF, CF, DJMM, SV1} at least when base field is the field of complex numbers.
For $\Char \K =p$ and $h  \not\in  \bF_p$, we prove the following version of the steepest descent result for $E_h$\,.\ Let  $\text{Crit}_P\mono \bP^1_S$ be
the critical locus  of the polynomial $P(x)$ and  $S^\circ \subset S$
 the open subset over which the projection $\pi: \text{Crit}_P \to S$ 
 is \'etale.  In Theorem \ref{th:main_char_p_generic_h} we show that $E_h$ over $S^\circ$ is isomorphic  
to $\cP^{h}|_{\pi^{-1}(S^\circ)}$ (viewed as an $\cO_{S^\circ}$-module with a connection). Then the irreducibility result, Theorem \ref{thm:intro_irr},  is derived from geometric properties of the finite map $\pi$,
 see Lemmas \ref{lem:irred_crit_locus} and \ref{lem:beta_anzatz}. 

 For $h \in  \bF_p$, we identify in Section \ref{sec 4} the local system  $E_h$ with  an isotypic component of the de Rham cohomology (with constant coefficients) of a smooth projective curve $X$ over $S$ acted upon by a finite group. The flat subbundle $\mc U$ from Theorem \ref{thm:main3_introduction} is constructed using the conjugate filtration on the de Rham cohomology and its Lagrangian property. 
 Finally, the description of the flat sections of $E_h$ is proven using a result of Katz relating the p-curvature of the Gauss-Manin connection and the Kodaira-Spencer operator and using the key non-degeneracy property of the latter,
 see Corollary \ref{cor 5.3} and Lemma \ref{lm:image_of_p_curvature_is_big}.

\subsection{Organization of the paper}

In Section \ref{sec 2} we study the cohomology of the logarithmic local system $\mc P^h$ over
the projective line  $\bP^1$. This local system is associated with the master function  $P^h$.   
In Section \ref{sec 3} we study  $\mc P^h$ in characteristic $p$. This section contains proofs of all our main results except for the $p$-curvature formula in 
Theorem \ref{thm Pcu}. 
In Section \ref{sec 4} we introduce the family of algebraic curves defined by the affine equation
$y^q = \prod_{i=1}^n (x-z_i)$
and depending on parameters $z$. We 
study the cohomology of these curves in characteristic $p$.
In Section \ref{sec 5.1}, we give a different proof of Theorem \ref{thm:main3_introduction},
 and in Section \ref{sec p-cu}  we prove our last main result -- Theorem \ref{thm pPk}, which described the $p$-curvature operators.

\subsection*{Acknowledgments}
 
The authors thank Pavel Etingof and Evgeny Mukhin  for discussions and Andrey Gabrielov for providing a proof of Lemma
\ref{lem G} over the field $\C$.
The first author would like to extend his gratitude to IHES for the hospitality provided during the development of this paper in July 2023.

\section{Logarithmic local system over $\bP^1$ and its cohomology}
\label{sec 2}
In this section we explain a geometric interpretation
of the KZ  connection as the Gauss-Manin connection and prove a few elementary results not specific to the positive characteristic case.
\subsection{Local system $\cP^h$ and its logarithmic de Rham cohomology}\label{ss:loc_system}
Fix an integer $n>1$. Let $S\subset \bA^n= \Spec\bZ[z_1, \dots, z_n]$
be the complement to the union of hyperplanes $z_i - z_j=0$, 
$1\leq i, j\leq n$, that is,
$$
S= \Spec \bZ [z_i, (z_i-z_j)^{-1},\; 1\leq i\ne j \leq n].
$$
Set $\bP^1_S=\bP^1 \times S$. We denote by $x$ the coordinate on $\bA^1 \subset \bP^1$. 

Define divisors $T_i \mono \bA^1_S\subset  \bP^1_S$, $1\leq i \leq n$, by equations $x=z_i$. Set 
\bea
T_\infty = \{\infty\} \times S \mono \bP^1_S, \quad T= T_\infty \cup \bigcup_{i=1}^n T_i.
\eea
Note that divisors $T_i$, $T_\infty$ are pairwise disjoint and each of them  projects isomorphically to $S$. 

We denote by $\Omega^m_{\bP^1_S}(\log T)$ the sheaf of logarithmic differential form on $(\bP^1_S, T)$. Thus,
$\Omega^m_{\bP^1_S}(\log T)$ is an $\cO_{\bP^1_S}$-submodule of
$\Omega^m_{\bP^1_S}\otimes \cO(\infty[T])$\footnote{ Where $\cO(\infty[T])$ stands for the sheaf of rational functions
regular outside of $T$,},
 whose sections on 
an open subset $U\subset \bP^1_S $  
are differential forms $\mu$ on $U\setminus (U\cap T)$ such that $f \mu$ and  $f d \mu$ are both
regular on $U$. Here $f=0$ is a local equation of the divisor $T$.

Define a closed logarithmic $1$-form on  $\bP^1_S$ by the formula 
\begin{equation}\label{log-1-form}
\eta=   \frac{dP}{P} \in \Gamma(\bP^1_S, \Omega^m_{\bP^1_S}(\log T)),
\qquad  P= \prod_{i=1}^n (x-z_i). 
\end{equation}
Define a logarithmic de Rham local system $\cP^h$ on  $\bP^1_S$:
$\cP^h = \cO_{\bP^1_S}[h, h^{-1}]$ as an $\cO_{\bP^1_S}$-module;
  the logarithmic connection is given by the formula:
\begin{equation}\label{eq:logconnectionh}
 \cO_{\bP^1_S}[h, h^{-1}] \rar{\nabla^{\cP^h}}  \Omega^1_{\bP^1_S}(\log T)[h, h^{-1}], \quad  \nabla^{\cP^h}(f) = df + h f \eta. 
\end{equation}
Since $\eta$ is closed  the connection  $\nabla^{\cP^h}$ is flat.
\begin{rem}
   Our notation for the local system is inspired by the formula:
$$
d(P^h f) =  P^h\Big(df +h f \frac{dP}{P}\Big).
$$ 
\end{rem}
Let $ H^\bullet_{dR, \log}(\bP^1_S/S,\cP^h) $ be the relative logarithmic de Rham cohomology of $\bP^1_S$  with coefficients in $\cP^h$.
Recall that the latter is defined to be the hypercohomology of $\bP^1_S$ with coefficients in the complex \eqref{eq:logconnectionh}

Since the higher cohomology groups of $\cO_{\bP^1_S}$ and $\Omega^1_{\bP^1_S/S}(\log T)$ 
vanish,
$ H^\bullet_{dR, \log}(\bP^1_S/S,\cP^h) $ is just the cohomology of the complex 
\begin{equation}\label{log_de_Rham_P^h_explicit}
\Gamma (\bP^1_S,  \cO_{\bP^1_S})[h, h^{-1}] 
\ \xrightarrow{\nabla^{\cP^h}}\  \ \Gamma(\bP^1_S, \Omega^1_{\bP^1_S/S}(\log T))[h, h^{-1}].
\end{equation}
Observe, that $ \Gamma(\bP^1_S, \Omega^1_{\bP^1_S/S}(\log T))$ is a free $\Gamma (S, \cO_S)$-module  on 
$$\eta_i =\frac{d(x-z_i)}{x-z_i}, \quad 1\leq i \leq n.$$
Thus, complex (\ref{log_de_Rham_P^h_explicit}) has the  form
\begin{equation}\label{log_de_Rham_P^h_more_explicit}
\fO \  \to \  
    \bigoplus_{i=1}^n  \fO \cdot \eta_i\,, \quad f\mapsto   h \sum_i f \eta_i\,, 
\end{equation}   
where  $$\fO= \bZ [h, h^{-1}, z_i, (z_i-z_j)^{-1},\; 1\leq i\ne j \leq n].$$
In particular, we infer the following.
\begin{pr}\label{pr:basis_for_E} The following statements are true.
\begin{enumerate}
    \item  $ H^0_{dR, \log}(\bP^1_S/S,\cP^h)=0 $.
    \item
$E:=H^1_{dR, \log}(\bP^1_S/S,\cP^h) $
is a free $\fO$-module generated by  classes $[\eta_i]$
with a single relation 
$$\sum_{i=1}^{n}\, [\eta_i]=0.$$
\end{enumerate}
\end{pr}
The relative logarithmic de Rham cohomology is equipped with the Gauss-Manin connection $\nabla^E$. We refer the reader to \cite[Section 1.4]{katz} for a definition of the Gauss-Manin connection on the logarithmic de Rham cohomology. 
An explicit formula for the connection on $E$ is
given in Section \ref{ss:explicitformula}.

\subsection{Calibration}\label{ss:calibration} 

For a Cartier divisor $D$ 
on a scheme  $X$ we denote by $\cO_X(D)$ the corresponding invertible sheaf.
For an $\cO_X$-module $\cF$, we shall write $\cF(D)$ for $\cF \otimes _{\cO_X}\cO_X(D)$.

Let $m_\infty$, $m_1, \dots ,m_n$  be integers. Denote by 
$$C_{m_\infty, m_1, \dots ,m_n}\ :\   C_{m_\infty, m_1, {\dots} ,
m_n}^0 \ \to\   C_{m_\infty, m_1, {\dots} ,m_n}^1 $$ the complex 
$$
\cO_{\bP^1_S}\Big((m _\infty-1) T_\infty  + \sum_{i=1}^n (m_i-1)T_i\,\Big)[h, h^{-1} ] \
 \xrightarrow{\nabla ^{\cP^h}}  \ \Omega^1_{\bP^1_S/S}\Big(\, m _\infty T_\infty  +  \sum_{i=1}^n  m_iT_i \,\Big)[h, h^{-1}]$$
of sheaves on $\bP^1_S$. The differential  is well-defined because $\nabla ^{\cP^h}= d +h  \frac{dP}{P}$ and  $\frac{dP}{P}$ is a section of 
  $\Omega^1_{\bP^1_S/S}(T_\infty  +  \sum_{i=1}^n  T_i \,)$.
  Note that $C_{1, 1, {\dots} , 1}$ is the relative logarithmic de Rham complex \eqref{eq:logconnectionh}.
  \begin{lem}\label{lm:divisage}    

${}$

\begin{itemize}
\item[(i)]  Fix an integer  $1\leq i \leq n$. 
Then the embedding of complexes $C_{m_\infty, m_1, \dots ,m_n}  \mono C_{m_\infty, m_1, \dots ,m_i+1, \dots, m_n}$ induces a quasi-isomorphism after inverting $m_i- h \in \bZ[h]$:
$$C_{m_\infty, m_1, \dots ,m_n}[(m_i- h)^{-1}]   \iso C_{m_\infty, m_1, \dots ,m_i+1, \dots, m_n}[(m_i- h)^{-1}].$$
\item[(ii)] 
The embedding of complexes $C_{m_\infty,  m_1, \dots ,m_n}  \mono C_{m_\infty +1,  m_1, \dots ,m_n}$ 
induces a quasi-isomorphism after inverting $m_{\infty}+ n h \in \bZ[h]$.
\end{itemize}
\end{lem}
\begin{proof} For part  (i) we have to show that the complex
\begin{equation}
\label{eq:proofquot}
 C_{m_\infty, m_1, \dots ,m_i+1, \dots, m_n}^0/C_{m_\infty,  m_1, \dots ,m_n}^0
 \ \rar{\nabla ^{\cP^h}} 
\ C_{m_\infty, m_1, \dots ,m_i+1, \dots, m_n}^1/C_{m_\infty,  m_1, \dots ,m_n}^1
 \end{equation}
 becomes acyclic after inverting $m_i- h$.
 Note that the quotient sheaves in (\ref{eq:proofquot}) are supported on $T_i$. 
 The  restriction of 
$C_{m_\infty, m_1, \dots ,m_i+1, \dots, m_n}^0/C_{m_\infty,  m_1, \dots ,m_n}^0 $
to $T_i$ 
 is a free $\cO_{T_i}[h, h^{-1} ]$-module on 
$\frac{1}{(x-z_i)^{m_i}}$. The  restriction of
 $C_{m_\infty, m_1, \dots ,m_i+1, \dots, m_n}^1/C_{m_\infty,  m_1, \dots ,m_n}^1$
 is a free $\cO_{T_i}[h, h^{-1} ] $-module on 
$\frac{d(x-z_i)}{(x-z_i)^{m_i+1}}$.

Note that 
$$
\nabla ^{\cP^h}\left(\frac{1}{(x-z_i)^{m_i}}\right)= 
d\left(\frac{1}{(x-z_i)^{m_i}}\right) + \frac{h}{(x-z_i)^{m_i}} \eta =
(h -m_i) \frac{d(x-z_i)}{(x-z_i)^{m_i+1}} 
$$ 
in the 
quotient $C_{m_\infty, m_1, \dots ,m_i+1, \dots, m_n}^1/C_{m_\infty,  m_1, \dots ,m_n}^1$.
Thus, complex (\ref{eq:proofquot}) is isomorphic to the complex
$$
\cO_{T_i}[h, h^{-1} ]  \ \xrightarrow{ (h -m_i) \Id}\  \cO_{T_i}[h, h^{-1} ] \,. 
$$
This proves part (i). For part (ii) the 
argument is similar using that the residue of $\eta$ at $T_\infty$ is  $-n$. 
\end{proof}

\subsection{Log de Rham cohomology vs. ordinary}
\label{ss:basechange}

Let $R$ be a commutative ring equipped with a homomorphism $$s: \bZ[h, h^{-1}] \to R.$$ 
We denote by $\bar h \in R^\times $ the image of $h$. 
Since the cohomology  $ H^\bullet_{dR, \log}(\bP^1_S/S,\cP^h) $ is flat over $\bZ[h, h^{-1}]$
the base change morphism is an isomorphism:
$$H^\bullet_{dR, \log}(\bP^1_S/S,\cP^h) \otimes _{\bZ[h]} R \ 
\iso \
H^\bullet_{dR, \log}(\bP^1_{S},\cP^h \otimes _{\bZ[h]} R).$$
Let $U$ be the complement  to $T$ in $\bP^1_{S}$. Consider the restriction morphism
\begin{equation}
\label{mor:log_to_open}
 H^\bullet_{dR, \log}(\bP^1_{S}/S,\cP^h \otimes _{\bZ[h]} R) \ 
  \to \
  H^\bullet_{dR}(U/S, (\cP^h \otimes _{\bZ[h]} R)|_{U}).
    \end{equation}
    Recall that (\ref{mor:log_to_open}) commutes with the Gauss-Manin connection.
\begin{pr}\label{pr:log_to_open}
 Assume that for every  positive integer $m$  the 
 elements $\bar h-m$, $n\bar h+m$ are invertible in $R$.
 Then morphism (\ref{mor:log_to_open}) is an isomorphism.
\end{pr}
\begin{proof} 
Using Lemma \ref{lm:divisage} the morphism 
$$ 
C_{1, 1, {\dots},1} 
 \otimes _{\bZ[h]} R \  \to\  \rightlimit m (C_{m, m, {\dots} ,m}  \otimes _{\bZ[h]} R)
 $$ 
is a quasi-isomorphism. 
To complete the proof it remains to observe that
$$
R^\bullet \Gamma (\bP^1_{S}, \rightlimit m (C_{m, m, \dots,
m}  \otimes _{\bZ[h]} R))\iso
H^\bullet_{dR}(U/S, (\cP^h \otimes _{\bZ[h]} R)|_{U}).
$$
\end{proof}

\subsection{Local system $\cP^{h+1}$}

Define a logarithmic local system $\cP^{h+1}$: $\cP^{h+1} = \cO_{\bP^1_S}[h, h^{-1}]$ as an $\cO_{\bP^1_S}$-module. The logarithmic connection is given by the formula:
\begin{equation}\label{eq:logconnectionh+1}
 \cO_{\bP^1_S}[h, h^{-1}] \ \xrightarrow{\nabla^{\cP^{h+1}}}\ \ \Omega^1_{\bP^1_S}(\log T)[h, h^{-1}],
 \quad  \nabla^{\cP^{h+1}}(f) = df + 
 (h+1) f \eta. 
\end{equation}
Note that the restrictions of $\cP^{h}$, $\cP^{h+1}$ to $U= \bP^1_S\setminus T $ are isomorphic:
\begin{equation}\label{eq:h+1_to_h}
\cP^{h+1}\big|_{U}\ \iso\ \cP^{h}\Big|_{U},
 \quad f\mapsto f P. 
\end{equation}
Observe that multiplication by $P$ yields a commutative diagram 
\begin{equation}\label{dia:h+1_to_h}
\def\normalbaselines{\baselineskip20pt
\lineskip3pt  \lineskiplimit3pt}
\def\mapright#1{\smash{
\mathop{\to}\limits^{#1}}}
\def\mapdown#1{\Big\downarrow\rlap
{$\vcenter{\hbox{$\scriptstyle#1$}}$}}
\begin{matrix}
\cO_{\bP^1_S} [h, h^{-1}]& \xrightarrow{\nabla ^{\cP^{h+1}}}  &\Omega^1_{\bP^1_S/S}(\log T)[h, h^{-1}]  \cr
 \mapdown{P}  &  &\mapdown{P}  \cr
\cO_{\bP^1_S}(-\Div P) [h, h^{-1}] &  \xrightarrow{\nabla ^{\cP^{h}}} &  \Omega^1_{\bP^1_S/S}(T -\Div P)[h, h^{-1}],
\end{matrix}
\end{equation}
where the vertical arrows are isomorphisms. 
This defines an isomorphism 
\begin{equation}\label{dia:h+1_to_hbis} H^\bullet_{dR, \log}(\bP^1_S/S,\cP^{h+1}) \iso 
R^\bullet \Gamma ( \bP^1_S, C_{n+1, 0, \dots ,0}),
\end{equation}
where the complex  $C_{n+1, 0, \dots ,0}$ is defined in Section \ref{ss:calibration}.
\begin{thm}\label{Th:h_to_h+1}  Denote $a= \prod _{i=1}^{n}(i +nh) \in \bZ[h]$. Then
morphism (\ref{dia:h+1_to_hbis}) induces an isomorphism 
\begin{equation}
\label{formula:h+1_to_h}
H^\bullet_{dR, \log}(\bP^1_S/S,\cP^{h+1})[a^{-1}] 
\ \iso\  
H^\bullet_{dR, \log}(\bP^1_S/S,\cP^{h}) [a^{-1} ]
\end{equation}
of local systems on $S$.
\end{thm}

\begin{proof}
Consider the morphisms  of  complexes
$$ 
C_{n+1, 0, {\dots},0}   \ \leftarrow \ C_{1, 0, {\dots} , 0}\ \to\ C_{1, 1, {\dots} , 1}\,. 
 $$
The right arrow is a quasi-isomorphism by  part (i) of Lemma \ref{lm:divisage}. 
The second morphism becomes a 
quasi-isomorphism after inverting $a$ by 
 part  (ii) of Lemma \ref{lm:divisage}. Since   $C_{1, 1, {\dots}
  , 1}$ is the relative logarithmic de Rham complex of $\cP^h$, we infer
 $$
 R^\bullet \Gamma ( \bP^1_S, C_{n+1, 0, {\dots} ,0})[a^{-1}] 
 \iso H^\bullet_{dR, \log}(\bP^1_S/S,\cP^{h})[a^{-1}].
 $$
 Combining this with (\ref{dia:h+1_to_hbis}) we obtain (\ref{formula:h+1_to_h}).
 To complete the proof it remains to check that (\ref{formula:h+1_to_h}) is compatible with
 the Gauss-Manin connection. Pick an embedding  $\bZ[h] \mono \bC$.  Since both sides of (\ref{formula:h+1_to_h}) are free modules over $\bZ[h, (ah)^{-1}]$ is suffices to verify 
 the compatibility with the Gauss-Manin connection after the base change to $\bC$. 
 By Proposition \ref{pr:log_to_open} the restriction morphism
 $$
 H^\bullet_{dR, \log}(\bP^1_{S\otimes \bC}/ S\otimes \bC,\cP^h \otimes _{\bZ[h]} \bC)   \to 
  H^\bullet_{dR}(U\otimes \bC/ S\otimes \bC, (\cP^h \otimes _{\bZ[h]} \bC)|_{U\otimes \bC}) 
    $$
    is an isomorphism compatible with the Gauss-Manin connection. 
    Under this identification  (and the analogous isomorphism for $\cP^{h+1}$) isomorphism 
    (\ref{formula:h+1_to_h}) is induced by (\ref{eq:h+1_to_h}). The compatibility of 
     (\ref{formula:h+1_to_h}) with the Gauss-Manin connection follows.
 \end{proof}

\subsection{Explicit formula for Gauss-Manin connection on $E$}
\label{ss:explicitformula}

For $1\leq i,j \leq n$,  let $\Omega_{ij}$ be the $n\times n$ matrix defined in \eqref{Omij}. 
Set
 $$
 A\,=\, \sum_{1\leq i\ne j \leq n} \frac{\Omega_{ij}}{z_i - z_j}dz_i\, \in \,\Mat_{n\times n}\big(\,\Gamma(S, \Omega^1_S)\,\big).
 $$
 Consider  the connection on the free   
 $\Gamma(S, \cO_S)$-module  
 $$\fO^{\oplus n} =\bigoplus_{i=1}^{i=n} \fO \cdot e_i, \quad  \fO:= \Gamma(S, \cO_S)[h, h^{-1}]$$
 given by  the formula 
\begin{equation}\label{eq:explicit_formula_for_connection}
    \nabla= d - h A.
\end{equation}
One verifies that the connection is flat:
$$
d(A) - h A\wedge A =0.\,
\footnote{This is equivalent to $d(A)=A\wedge A=0$.}
$$

\begin{pr}\label{pr:expliciteformulafortheKZconn}
The morphism 
$$
(\fO^{\oplus n}, \nabla) \to (E, \nabla^E), \quad e_i\mapsto [\eta_i], \quad 1\leq i\leq n , 
$$
is flat.
\end{pr}
\begin{proof}
A direct computation of the Gauss-Manin connection on $E$ is involved: in general, the vector field $\frac{\partial}{\partial z_i}$ can not  be lifted to a vector field on $\bP^1_Z$ preserving the divisor $V$. However,  
arguing as in the proof of Theorem \ref{Th:h_to_h+1} it suffices to check
that the composition 
$$
(\fO^{\oplus n}, \nabla) \to (E, \nabla^E) \to (H^1_{dR}(U/S, \,\cP^h), \nabla^{GM}),
$$
is flat. 
Here the second arrow depicts  the restriction map (which is an isomorphism for generic $h$) and $\nabla^{GM}$ stands for the Gauss-Manin connection on $H^1_{dR}(U/S, \cP^h)$. Computing the latter is easy:
for every $i\ne j$, we have 
$$
\nabla^{GM}_{\!\!\frac{\partial}{\partial z_i}}([\eta_j])= \left[\left(\left(\frac{1}{x - z_j}\right)'_{z_i} +h \frac{P'_{z_i}}{(x-z_j)P}\right)dx\right]= \left[- h \frac{\eta_i - \eta_j}{z_i - z_j}dx \right].$$
Since $\sum_i [\eta_i] =0 $ the formula above determines the connection. 
\end{proof}

\subsection{Duality}
\label{ss:duality}

Denote by $\cP^{-h}$  (resp.  $\cP^{-h}(-T)$)  the logarithmic local system  $\cO_{\bP^1_S}[h, h^{-1}]$ (resp. $\cO_{\bP^1_S}(-T)[h, h^{-1}]$) endowed with 
a logarithmic connection 
\begin{equation}
\label{eq:logconnection-h-T}
  \nabla^{\cP^{-h}}(f) = df - h f \eta.
\end{equation}
The cup product induces a morphism of complexes 
\begin{equation}
    \begin{split}
  (\cO_{\bP^1_S}[h, h^{-1}]\rar{\nabla^{\cP^h}}  \Omega^1_{\bP^1_S/S}(\log T) [h, h^{-1}] )\otimes & (\cO_{\bP^1_S}(-T)[h, h^{-1}] \rar{\nabla^{\cP^{-h}}}  \Omega^1_{\bP^1_S/S}[h, h^{-1}])\\
\rar{} (\cO_{\bP^1_S}[h, h^{-1}]\rar{d}  \Omega^1_{\bP^1_S/S}[h, h^{-1}]) & 
\end{split}
\end{equation}
and, consequently, a pairing  
\begin{equation}\label{eq:logpoincare}
 H^1_{dR, \log}(\bP^1_S/S,\cP^h) \otimes 
 H^1_{dR, \log}(\bP^1_S/S,\cP^{-h}(-T))\to  H^2_{dR}(\bP^1_S/S) [h, h^{-1}]) \iso \fO
\end{equation}
compatible with the Gauss-Manin connection. 
Since
 $$H^0(\bP^1_S, \cO_{\bP^1_S}(-T))= H^0(\bP^1_S,  \Omega^1_{\bP^1_S/S}) =0$$
  the logarithmic de Rham cohomology $H^\bullet_{dR, \log}(\bP^1_S/S,\cP^{-h}(-T))$ is computed by the complex 
\begin{equation}\label{log_de_Rham_P^{-h}_explicit}
H^1(\bP^1_S, \cO_{\bP^1_S}(-T) [h, h^{-1}] )  \rar{\nabla^{\cP^{-h}}} H^1(\bP^1_S,  \Omega^1_{\bP^1_S/S}[h, h^{-1}]),
\end{equation}
supported in cohomological degrees $1$ and $2$.  
The Serre duality induces a  $\fO$-linear duality between complexes (\ref{log_de_Rham_P^h_explicit}) and  (\ref{log_de_Rham_P^{-h}_explicit}) which reduces to (\ref{eq:logpoincare}) in degree $1$. In the other words,  if $\eta_i^*$, $1\leq i \leq n$,  is a basis for 
$H^1(\bP^1_S, \cO_{\bP^1_S}(-T))$ dual to $\eta_i \in  \Gamma(\bP^1_S, \Omega^1_{\bP^1_S/S}(\log T)) $ and  $H^1(\bP^1_S,  \Omega^1_{\bP^1_S/S})$ is identified with $\Gamma(S, \cO_S)$ via the trace map, then complex (\ref{log_de_Rham_P^{-h}_explicit}) has the form
\begin{equation}
\label{log_de_Rham_P^_{-h}_more_explicit}
 \bigoplus_{i=1}^n \fO \cdot \eta_i^*  \ \to \   \fO,
 \quad
 \sum f_i \eta_i^* \mapsto  - h \sum_i f_i\,.
\end{equation}  
In particular, we see that  $H^1_{dR, \log}(\bP^1_S/S,\cP^{-h}(-T))$ is a free  $\fO$-module and (\ref{eq:logpoincare}) is a perfect pairing.

\smallskip

Using Lemma \ref{lm:divisage} the morphism 
\begin{equation}\label{eq:another_calibration}
\gamma: \ H^1_{dR, \log}(\bP^1_S/S,\cP^{-h}(-T))\ \to\ H^1_{dR, \log}(\bP^1_S/S,\cP^{-h}),
\end{equation}
induced by $\cP^{-h}(-T) \to \cP^{-h}$,  is an isomorphism after inverting $n$ and $h$.
\begin{lem} For every $1\leq i, j \leq n$, we have that 
$$\gamma(\eta_i^* - \eta_j^*) = -h([\eta_i] - [\eta_j]).$$
\end{lem}
Proof is left to the reader.
\begin{cor}\label{cor:duality_section2}
The logarithmic Poincare duality  (\ref{eq:logpoincare}) and $\gamma$ induce a perfect pairing 
\begin{equation}\label{eq:paringP^handP^{-h}}
H^1_{dR, \log}(\bP^1_S/S,\cP^h)[n^{-1}] \otimes H^1_{dR, \log}(\bP^1_S/S,\cP^{-h})[n^{-1}]\to \fO[n^{-1}]
\end{equation}
compatible with the Gauss-Manin connection and given by the formula:
\begin{gather*}
[\eta_i]  \otimes [\eta_j] \mapsto -h^{-1}(\delta_{ij} - \frac{1}{n}), \quad \;
1\leq i, j \leq n.
\end{gather*}
\end{cor}

\section{Study of $\cP^h$ in characteristic $p>0$}
\label{sec 3}

For the duration of this section, let $p$ be a prime integer that does not divide $n$. We denote by 
$\cP^h \otimes \bF_p$ 
 (resp. $E \otimes \bF_p$) the restriction of  $\cP^h$ along the embedding $ \bP^1_{S_{\bF_p}}: = \bP^1_{S}\times \Spec \bF_p \mono \bP^1_{S}$  (resp. $S_{\bF_p}: = S\times \Spec \bF_p \mono S$). 
Note that $E\otimes  \bF_p$ can be also interpreted  as
the relative logarithmic de Rham cohomology of $\bP^1_{S}$
with coefficients in  $\cP^h\otimes \bF_p$.

Recall that $E\otimes  \bF_p$  is a module over $\bF_p[h, h^{-1}]$. In Sections
\ref{ss:cohomology_of_E}, \ref{ss:critical_locus}, \ref{ss:irreducibility} we study the local system
\begin{equation}
    E^\flat :=  E\otimes_{ \bZ[h]}  \bF_p[h, (h^p-h)^{-1}]
\end{equation}
obtained from $E\otimes  \bF_p$ by inverting $h^{p-1} -1$.
In Section \ref{ss:cohomology_of_E} we show that $E^\flat$ has trivial de Rham cohomology in all degrees. In particular, $E^\flat$ has no non-zero flat sections.
In Section 
\ref{ss:critical_locus} we consider the critical locus $ \text{Crit}_P \subset \bA^1_{S}$ of the function $P(x)=\prod_i(x-z_i)$ and the open subset $S^\circ \subset S$ over which the  projection
$\nu: \text{Crit}_P \to S$ is a (finite) \'etale map.
We show that, for $p>2$, the restriction of $E^\flat$ to  the non-empty open subset $S^\circ_{\bF_p}  \subset S_{\bF_p}$ 
is isomorphic to the restriction of $\cP^h$ to $\nu^{-1}(S^\circ _{\bF_p}) \subset \text{Crit}_P \times \Spec \bF_p$, that is the relative de Rham cohomology of the $1$-dimensional family 
$\bP^1 \times S^\circ _{\bF_p} \to S^\circ _{\bF_p}$ is isomorphic to the de Rham cohomology of the finite \'etale map $\nu^{-1}(S^\circ _{\bF_p})\to S^\circ _{\bF_p}$. 
In Section
\ref{ss:irreducibility} we apply the above result to prove the irreducibility of $E^\flat$ for $n<p$.

In Section
\ref{sec 3.6} we study the specialization of $E\otimes \bF_p$ to $h\in \bF_p$. In particular, we prove Theorem \ref{thm:main3_introduction}.

The notion of the $p$-curvature and (generalized)  Cartier isomorphism play a crucial role in our study. We review this material in Section
\ref{ss:review}.

\subsection{Review of $p$-curvature and Cartier isomorphism}
\label{ss:review}

 Recall from \cite[Section 3.0]{katz}  the notion of $p$-curvature of an $\cO_X$-module $M$ over a smooth scheme $X$ over a base $Y$ of characteristic $p$ equipped with an integrable connection $\nabla$ relative to $Y$. For any vector field $\theta\in \text{Der}(U/Y)$ on an open subset $U\subset X$ the $p$-curvature operator 
  $$\Psi_\theta: M|_{U} \to M|_{U}$$
  is defined by the formula 
  $$\Psi_\theta = \nabla_{\theta}^p - \nabla_{\theta^{[p]}},$$
  where $\theta^{[p]}$ is the  $p$-th iterate of the derivation $\theta$
  (which is again a derivation) and 
  $ \nabla_{\theta}$ is  the covariant derivative along $\theta$. One verifies that $\Psi_\theta$  
  is $\cO_U$-linear and flat with respect to $\nabla$. Moreover, one has
  $$\Psi_{\theta +\theta'} = \Psi_{\theta}+\Psi_{\theta'}, \quad \Psi_{\theta}\Psi_{\theta'}= \Psi_{\theta'}\Psi_{\theta},\quad  \Psi_{f\theta} = f^p \Psi_\theta\,. $$
  The $p$-curvature operators can be organized into a single $\cO_X$-linear map 
  \begin{equation}\label{def:p-curvature}
  \Psi: M \to F_{\text{abs}}^* \Omega^1_{X/Y}\otimes  M,  
   \end{equation}
  where $ F_{\text{abs}}: X\to X$ is the absolute Frobenius. Homomorphism (\ref{def:p-curvature}) is uniquely characterised by the property that, for every vector field $\theta \in  \text{Der}(U/Y)$ on an open subset $U\subset X$, the composition of 
$\Psi$  with the interior product $\iota_\theta$ is $\Psi_\theta$.

  We will use a generalization of the Cartier isomorphism for the de Rham cohomology with coefficients in a module  equipped with an integrable connection due to Ogus \cite{ogus}. Let us recall the statement. Given $(M, \nabla)$ as above consider the $F$-Higgs complex
  \begin{equation}\label{F-Higgs}
      M \to F_{\text{abs}}^* \Omega^1_{X/Y}\otimes  M  \to F_{\text{abs}}^* \Omega^2_{X/Y}\otimes  M \to \cdots 
  \end{equation}
  where the differential carries a germ $\omega \otimes v $ of 
  $F_{\text{abs}}^* \Omega^m_{X/Y}\otimes  M$ to the image of 
  $\omega \otimes \Psi (v)$ under the wedge product map
  $$F_{\text{abs}}^* \Omega^m_{X/Y}\otimes F_{\text{abs}}^* \Omega^1_{X/Y}\otimes  M
  \to F_{\text{abs}}^* \Omega^{m+1}_{X/Y}\otimes  M. $$
  Recall that, for every quasi-coherent $\cO_X$-module $\Omega$ on $X$, its Frobenius pullback
  $F_{\text{abs}}^* \Omega$ has a canonical integrable connection, called the Frobenius descent connection, which is uniquely characterized by the property that sections of $F_{\text{abs}}^{-1} \Omega \subset F_{\text{abs}}^* \Omega$  are flat.
  Define an integrable connection on  $F_{\text{abs}}^* \Omega^m_{X/Y} \otimes M$ to be 
  the tensor product of  the Frobenius descent connection on $F_{\text{abs}}^* \Omega^m_{X/Y}$ and $\nabla$ on $M$. Then
   the differential in (\ref{F-Higgs}) is flat.
 In particular, cohomology sheaves $\cH^m_\psi(X/Y, M)$ of
  (\ref{F-Higgs}) are endowed with a connection (that we also denote by $\nabla$). 
   Ogus proved in \cite[Theorem 1.2.1]{ogus} that the $p$-curvature of $\cH^m_\psi(M, \nabla)$ is zero and there is a canonical isomorphism:
  \begin{equation}\label{ogus_p_higgs}
C:  \cH^m_{dR}(X/Y, M) \iso \cH^m_\psi(X/Y, M)^{\nabla=0}.
\end{equation}
Here the superscript  $\nabla=0$ stands for the subsheaf of flat sections and  $\cH^m_{dR}(X/Y, M)$ is the de Rham cohomology sheaf (on $X$). 

Let us explain an explicit construction of the Cartier isomorphism (\ref{ogus_p_higgs}) for $X$ being an open subscheme of $ \bA_Y^1= Y \times \Spec \bF_p[x].$ For $m=0$ the Cartier operator carries a flat local section of $M$ to itself. For $m=1$ and a cohomology class represented by a local section $dx\otimes v$ of 
$\Omega^1_{X/Y}\otimes  M$, we have that
 \begin{equation}\label{ogus_explicit_formula}
C([dx\otimes v]) = - [dx \otimes  \Big(\nabla
\Big(\frac{\partial}{\partial x}\Big)\Big)^{p-1}(v)]. 
\end{equation}

  In the logarithmic context the $p$-curvature is defined similarly yielding  a complex  
  \begin{equation}\label{logHiggscomplex}
   M \to F_{\text{abs}}^* \Omega^1_{X/Y}(\log T)  \otimes M \to  F_{\text{abs}}^* \Omega^2_{X/Y}(\log T)  \otimes M \to \cdots
   \end{equation}
  see \cite[Section 3.1]{ogus}. There is also a logarithmic version of the generalized Cartier isomorphism:
  \begin{equation}\label{ogusmain}
C: \cH^m_{dR, \log}(X/Y, M) \iso \cH^m_{\psi, \log}(X/Y, M)^{\nabla=0};
\end{equation}
see \cite[Theorem 3.1.1]{ogus}.

\subsection{De Rham cohomology of $E \otimes \bF_p$}
\label{ss:cohomology_of_E}
Recall that $E\otimes  \bF_p$  is a module over $\bF_p[h, h^{-1}]$. Set
\begin{equation}
    E^\flat :=  E\otimes_{ \bZ[h]}  \bF_p[h, (h^p-h)^{-1}].
\end{equation}

\begin{thm}
 Assume that   $p$ does not divide $n$. Then, for every integer $m\geq 0$, we have 
 $$ H^m_{dR}(S, E^\flat)=0.$$
 \end{thm}
\begin{proof}
Set
\begin{equation}
   \cP^{h, \flat}: = \cP^h \otimes_{ \bZ[h]}  \bF_p[h, (h^p-h)^{-1}].
\end{equation}
 By Proposition \ref{pr:log_to_open} the restriction morphism
$$ H^\bullet_{dR, \log}(\bP^1_{S}/S, \cP^{h, \flat})\to 
H^\bullet_{dR}(U/S,  \cP^{h, \flat}|_{U}) $$
is an isomorphism. In particular,
$$   E^\flat \iso H^1_{dR}(U/S,  \cP^{h, \flat}|_{U}).$$
Since the relative de Rham cohomology groups  are trivial in all other degrees using
the Leray spectral sequence we conclude that
$$ H^m_{dR}(S, E^\flat) \iso H^{m+1}_{dR}(U,  \cP^{h, \flat}|_{U}).$$
Thus, it is enough to show that 
 \begin{equation}\label{eq:globalcohvan}
H^{\bullet}_{dR}(U,  \cP^{h, \flat}|_{U})=0
\end{equation}
in all degrees. We will do this using the generalized Cartier isomorphism.

  The $p$-curvature of the local system $\cP^h\otimes \bF_p $ can be described explicitly. Since $\cP^h\otimes \bF_p $ is trivial rank $1$ module over  $\cO_{\bP^1_S} \otimes \bF_p[h, h^{-1}]$,
  the morphism $\Psi$ is necessarily a multiplication by a section of 
  $$F_{\text{abs}}^* \Omega^1_{\bP^1_{S_{\bF_p}}}(\,\log (T_{\bF_p})\,) \otimes \bF_p[h, h^{-1}] =\cO_{\bP^1_{S_{\bF_p}}}[h, h^{-1}] \otimes_{F_{\text{abs}}}
  \Omega^1_{\bP^1_{S_{\bF_p}}}(\,\log (T_{\bF_p})\,). $$ 
  According to \cite[Prop. 7.2.2]{katz} this section is equal to
  \begin{equation}\label{eq:p-curvatureoflinebundle}
     (h^p-h) \otimes \eta. 
\end{equation}
  In local coordinates $x, z_1, {\dots,} z_n$ on $\bP^1_{Z\otimes \bF_p}$ the above formula reads as follows:
   \begin{equation}\label{eq:p-curvatureoflinebundle_expl}
  \Psi_{\!\!\frac{\partial}{\partial z_i}}= (h^p-h) \Big(P^{-1}\frac{\partial P}{\partial z_i} \Big)^p, 
  \quad \Psi_{\!\!\frac{\partial}{\partial x}}= (h^p-h)\Big( P^{-1}\frac{\partial P}{\partial x}\Big)^p.
  \end{equation}
  The key property of  $\eta$ that we need
  is that it
does not vanish on $U$. Indeed, the interior product of $\eta$
with the vector field  $\frac{\partial }{\partial z_i}$, $1\leq i\leq n$, is equal to 
$$P^{-1}\frac{\partial P}{\partial z_i} = \frac{1}{z_i-x},$$
which is an invertible function on $U$. It follows that the $F$-Higgs complex (\ref{F-Higgs}) of $\cP^{h, \flat}|_{U}$ is acyclic in all degrees. Indeed, the latter is obtained 
from the acyclic complex of $\cO$-modules
$$\cO_{U_{\bF_p}} \rar{\wedge \eta} \Omega^1 _{U_{\bF_p}}  \rar{\wedge \eta} \cdots $$
by pulling it back along $F_{\text{abs}}$ and then tensoring with  $\bF_p[h, (h^p-h)^{-1}]$.
Applying the result of Ogus (\ref{ogus_p_higgs}) 
we conclude that the de Rham cohomology sheaves  $\cH^\bullet_{dR}(U,  \cP^{h, \flat}|_{U})$ are trivial in all degrees.  
This proves the cohomology vanishing (\ref{eq:globalcohvan}) and, thus, completes the proof of the theorem. 
\end{proof}
\begin{cor}\label{cor:no_solutions} Let $p$ be a prime number that does not divide $n$, and let $R$ be a ring equipped with a homomorphism $\bF_p[h, (h^p-h)^{-1}] \to R.$
Then $E \otimes _{\bZ[h]} R$ has no non-zero flat sections.
    \end{cor}

This statement also can be deduced from \cite[Theorem 3.3]{EV}.

We remark that Corollary implies seemingly stronger statement:  $E \otimes _{\bZ[h]} R$ has non-zero flat sections over the formal neighborhood of any closed point of $S_{\bF_p}$.
Indeed, we have the following general result.
\begin{lem}\label{lm:base_change_to_disk}
    Let $T=\Spec A$ be a smooth algebra over a perfect field $\K$ of characteristic $p$, $a\in T(\K)$ a $\K$-point of $T$, $\mathtt{m}$ the corresponding maximal ideal of $A$,
    and let $\hat{A}$ be the $\mathtt{m}$-adic completion $A$, that is $\hat{A}= \limfrom \, A/\mathtt{m}^n$. Let $M$ be a $A$-module with an integrable connection $\nabla: M \to \Omega^1_{T} \otimes M$. Define a connection on $M \otimes _A \hat{A}$ by the Leibniz formula. Then the map\footnote{For a ring $B$ of characteristic $p$, we denote by $B^p$ the subring of $p$-th powers.}
    $$ M^{\nabla=0} \otimes_{A^p} \hat{A}^p  \to  (M \otimes _A \hat{A})^{\nabla =0}, \quad  v\otimes f \mapsto v\otimes f$$
    is an isomorphism.
\end{lem}
\begin{proof}
    Consider the left exact sequence  of $A^p$-modules $0\to M^{\nabla=0} \to M \to  \Omega^1_{T} \otimes M$. Since $\hat{A}^p$ is flat over $A^p$ the sequence remains exact after tensoring over $A^p$ with $\hat{A}^p$. It remains to observe that the product map $A\otimes_{A^p} \hat{A}^p \to \hat{A}$ is an isomorphism. 
\end{proof}
Using the above lemma and Corollary \ref{cor:no_solutions} we obtain the following.
\begin{cor}\label{cor:no_formal_solutions} Let $p$ be a prime number that does not divide $n$, and let $R$ be a ring equipped with a homomorphism $\bF_p[h, (h^p-h)^{-1}] \to R.$
Let $\K$ be a field of characteristic $p$, $\hat{A}$ be the formal completion of the local ring of $S_\K$ at a $\K$- point (thus, $\hat {A}$ is the ring of formal power series in $n$ variables with coefficients in $\K$).
Then $E\otimes _{\cO(S)} \hat {A} \otimes _{\bZ[h]} R$ has no non-zero flat sections.
    \end{cor}

\subsection{Geometric description of $ E^\flat$.}
\label{ss:critical_locus}

 Consider a closed subscheme $\text{Crit}_P$ of $\bA^1_{S}\subset \bP^1_{S}$ given by equation
$$P'_x(x, z_1, \dots,  z_n)=0.$$
Here $$P'_x= nx^{n-1} - (n-1) (\sum_{i=1}^n z_i) \, x^{n-2} +\dots $$ stands for the partial derivative of the polynomial 
$P= \prod_{i=1}^n(x-z_i)$ with respect to the  $x$-variable. 
The projection $$\nu: \text{Crit}_P \times \Spec \bZ[n^{-1}] \to S\times \Spec \bZ[n^{-1}]$$
is a finite flat morphism of degree $n-1$.
Let $S^\circ \subset S  \times \Spec \bZ[n^{-1}] $ be the maximal open subset of 
$S\times \Spec \bZ[n^{-1}] $ over which $\nu$ is \'etale. 
Explicitly, $S^\circ $  is given inside $ S  \times \Spec \bZ[n^{-1}]$ by the inequality 
$$\text{Disc}(P'_x)\ne 0,$$
where $\text{Disc}(P'_x)$ is the discriminant of the polynomial $P'_x$ in the $x$-variable.
In particular, $S^\circ $ is an affine scheme. 
Set $\text{Crit}_P^\circ = \nu^{-1}(S^\circ)$. Thus, we have a finite \'etale map
$$\nu: \text{Crit}_P^\circ \to S^\circ .$$
Set
 $$\cL = \nu_* \big( \cP^h\big|_{\text{Crit}_P^\circ}\big).$$ 
This is a locally free 
$\cO_{S^\circ}[h, h^{-1}]$-module of rank $n-1$ and, since  $\nu$ is 
 \'etale,  $\cL$ inherits an integrable connection $\nabla^\cL$.
For brevity we shall not distinguish between the category of local systems on a smooth  affine scheme and the category of modules over its ring of functions equipped with an integrable connection. 
 Under this equivalence, we have that 
 $$\cL = (\Gamma(\text{Crit}_P^\circ,\cO_{\text{Crit}_P^\circ})\otimes \bZ[h, h^{-1}], \nabla^\cL),$$
 where the connection $\nabla^\cL$ is given by the formula
\begin{equation}\label{eq:connection_on_L}
\nabla^\cL_{\!\!\frac{\partial}{\partial z_i}}(f)= \text{Lie}_{\theta_i} f +h
 \frac{\text{Lie}_{\theta_i} P}{P}f. 
 \end{equation}
 Here $\theta_i$ is the vector field on $\text{Crit}_P^\circ$ lifting $\frac{\partial}{\partial z_i}$ on $Z^\circ$. Explicitly, we have that 
$$\theta_i= \frac{\partial}{\partial z_i} -
\frac{P^{\prime \prime }_{xz_{i}}}{P^{\prime \prime }_{xx}} \frac{\partial}{\partial x}.$$

Let $p$ be an odd prime integer that does not divide $n$\footnote{Note that these assumptions guarantee that $S^\circ \otimes \bF_p$ is not empty: there exists a polynomial
$f(x)\in \bar{\bF}_p[x]$ such that neither $f(x)$ nor $f'(x)$ have multiple roots. However, if $p=2$ the scheme $S^\circ \otimes \bF_p$ is empty: every root of $f'(x)$ is a multiple root.   }.
Define a homomorphism of local systems
\begin{equation}\label{eq:critical_locus_and_E}
   (E|_{S^\circ} \otimes \bF_p, \nabla^E) \to (\cL\otimes \bF_p, \nabla^\cL)
\end{equation}
as follows. Let $W \subset \bA^1_{S^\circ} $
be any affine open neighborhood $\text{Crit}_P^\circ $ not intersecting the divisors $(x-z_i)$, $1\leq i \leq n$.
For any $1$-form 
$$fdx \in \Gamma (W, \Omega^1_{W/{S^\circ}} \otimes \bF_p)$$
set
\begin{equation}\label{eq:anotherdef}
C(fdx) = - \left(\left(\nabla^{\cP^h}_{\!\!\frac{\partial}{\partial x}}
\right)^{p-1}(f)\right)|_{\text{Crit}_P^\circ \times \Spec \bF_p}.
\end{equation}
Using the $p$-curvature formula (\ref{eq:p-curvatureoflinebundle_expl}), we have 
$$C(\nabla^{\cP^h}_{\!\!\frac{\partial}{\partial x}}(f) dx)
= -(h^p -h) \big( ( P^{-1}P'_x)^p f \big)|_{\text{Crit}_P^\circ \times \Spec \bF_p}=0$$
since $P_x'$ is equal to $0$ on $\text{Crit}_P$.
Thus, the Cartier morphism $C$ factors through the  relative de Rham cohomology:
\begin{equation}\label{eq:critical_locus_and_E_bis}
C\colon  H^1_{dR}(W /S^\circ,  \,\cP^h|_{W}\otimes \bF_p) \to \cL\otimes \bF_p.
\end{equation}
Morphism (\ref{eq:critical_locus_and_E}) is defined to be the composition of the restriction map on the de Rham cohomology ({\it cf.} (\ref{mor:log_to_open})) and the Cartier map 
 (\ref{eq:critical_locus_and_E_bis}).
The flatness of morphisms (\ref{eq:critical_locus_and_E_bis}) and (\ref{eq:critical_locus_and_E}) is immediate: the Gauss-Manin connection on the relative de Rham cohomology is given by the formula
$$
\nabla_{\!\!\frac{\partial}{\partial z_i}} [fdx]= \left[\nabla_{\!\!\frac{\partial}{\partial z_i}}^{\cP^h}(f)dx\right].
$$
Since $\nabla_{\!\!\frac{\partial}{\partial z_i}}^{\cP^h}$ commutes with $\nabla^{\cP^h}_{\!\!\frac{\partial}{\partial x}}$ it also commutes with $C$.
\begin{thm}\label{th:main_char_p_generic_h} Morphism (\ref{eq:critical_locus_and_E}) is an isomorphism after inverting  $h^p-h$:
$$ E^\flat |_{S^\circ}  \iso  \cL \otimes_{ \bZ[h]}  \bF_p[h, (h^p-h)^{-1}].$$
\end{thm}
\begin{proof} 
Let $W\subset U$ be the preimage of $S^\circ$.  Using Proposition (\ref{pr:log_to_open}) it is suffices to check that morphism 
(\ref{eq:critical_locus_and_E_bis}) is an isomorphism.
  We apply the generalized Cartier isomorphism   (\ref{ogus_p_higgs}), (\ref{ogus_explicit_formula}) to $X=W\times \Spec \bF_p $, $Y= S^\circ \times \Spec \bF_p$, $M=\cP^h|_{W\times \Spec \bF_p}$. The sheaves $\cH^\bullet_\psi(X/Y,M)$ are the cohomology sheaves of the complex
  $$
  \cO_{W}\otimes  \bF_p[h, (h^p-h)^{-1}] \,  \xrightarrow{(h^p-h) ( P^{-1}P'_x)^p} \,\cO_{W}\otimes  \bF_p[h, (h^p-h)^{-1}].
  $$
   Since $P^{-p}$ is invertible on  $W$ we conclude that 
  $\cH^0_\psi(X/Y,M)=0$ and $$\cH^1_\psi(X/Y,M)\iso \cO_{W}/\big((P'_x)^p\big) \otimes  \bF_p[h, (h^p-h)^{-1}] .$$ 
  Comparing formulae (\ref{ogus_explicit_formula}) and (\ref{eq:anotherdef}) we reduce the theorem 
  to the following assertion: the projection
 $$ \cO_{W}/\big((P'_x)^p\big)  \to \cO_{W}/\big(P'_x\big) $$
  induces an isomorphism
  \begin{equation}\label{eq:isoresflat}
      \Big(\cO_{W}/\big((P'_x)^p\big) \otimes  \bF_p[h, (h^p-h)^{-1}]\Big)^{\nabla^{\cP^h}=0}\
       \iso\
        \cO_{W}/\big(P'_x\big) \otimes  \bF_p[h, (h^p-h)^{-1}].  
  \end{equation}
  The expression at the left-hand side of (\ref{eq:isoresflat}) stands for the sheaf of sections of
  $\cO_{W}/\big((P'_x)^p\big) \otimes  \bF_p[h, (h^p-h)^{-1}]$ which are flat with respect to vector fields tangent to the fibers of the
  projection $W \to  S^\circ,$ { i.e.},  sections annihilated by $\nabla^{\cP^h}_{\!\!\frac{\partial}{\partial x}}$. As a connection relative to $S^\circ \times \Spec \bF_p$, the 
  $\nabla^{\cP^h}$ has  zero $p$-curvature on $\cO_{W}/\big((P'_x)^p\big)  \otimes  \bF_p[h] $: the image of $(h^p-h)( P^{-1}P'_x)^p$ in the quotient 
  $\cO_{W}/\big((P'_x)^p\big) \otimes  \bF_p[h]$ is equal to $0$. In addition, the relative spectrum of $\cO_{W}/\big(P'_x\big) \otimes  \bF_p$,
that is $ \text{Crit}_P^\circ \times \Spec \bF_p$, is \'etale over $S^\circ \times \Spec \bF_p.$ 
These two properties combined with the Cartier descent theorem (\cite{katz70}) imply   (\ref{eq:isoresflat}). Alternatively, one can complete the proof of Theorem \ref{th:main_char_p_generic_h}
 using the following lemma.

\begin{lem}
  Let $R$ be an $\bF_p$-algebra, $M$ a module over the polynomial algebra $R[x]$ with a connection
  $$\nabla_{\!\!\frac{\partial}{\partial x}}: M \to M.$$
  Assume that 
$$  \Big(\nabla_{\!\!\frac{\partial}{\partial x}}\Big)^p M = x^p M=0.$$
 Then the projection
 $$M^{\nabla =0} \to M/xM$$
 is an isomorphism.
\end{lem}

 \begin{proof} 
 
 We construct explicitly the inverse map. 
    The map $$-\Big( \nabla_{\!\!\frac{\partial}{\partial x}}\Big)^{p-1} x^{p-1}: M \to M$$
    factors though $M/xM$ and lands in $M^{\nabla =0}$. The induced morphism $ M/xM \to M^{\nabla =0}$ is the inverse to the projection. 
\end{proof}
\end{proof}

\subsection{Irreducibility of $E$}\label{ss:irreducibility}

As an application of Theorem \ref{th:main_char_p_generic_h} we shall prove  
the  irreducibility of $E$. Fix an algebraically closed field $\K$ of characteristic $p$  and 
an element $\bar{h}\in \K$. Consider the specialization $E_{\bar {h}}= E \otimes_{\bZ[h, h^{-1}]} \K$ viewed as a module over the ring $D_S$ of differential operators on $S$: for each $i$, the operator $\frac{\partial}{\partial z_i}$ acts by $\nabla^{E}_{\!\! \frac{\partial}{\partial z_i}}$. Let 
$$E_{\bar {h}, \eta}= E_{\bar {h}} \otimes_{\cO(S)\otimes \bK} \K(z)$$  
be the restriction of  $E_{\bar {h}}$ to the generic point of $S_\K$. By construction, $E_{\bar {h}}$ is an $(n-1)$-dimensional vector space over the field of rational function 
$\K(z)$ on $S_\K$ equipped with a $\K$-linear action of $D_S$. 
\begin{thm}\label{thm:irreducibility} Let  $p>n$ be an odd prime number.
  Assume that $\bar{h}\in \K\backslash \bF_p$. Then   $E_{\bar {h}, \eta}$ is an irreducible  $D_S\otimes \K$-module. 
\end{thm}

\begin{proof}
    We shall prove a stronger assertion. Let $A\subset D_S \otimes \K$ be the (commutative) $\K$-subalgebra generated by $z_i$ and  
    $\big( \frac{\partial}{\partial z_i}\big)^p$, $1\leq i\leq n$.  We will show that $E_{\bar {h}, \eta}$ is an irreducible $A$-module. Using 
    Theorem \ref{th:main_char_p_generic_h}, we identify the $\K(z)$-vector space  $E_{\bar {h}, \eta}$ with the space of functions $\L:=\cO( \text{Crit}_P^\circ)\otimes_{\cO(S^\circ)} \K(z)$ 
    on the generic fiber of the etale map 
    \begin{equation}\label{eq:etale_map}
      \nu: \text{Crit}_P^\circ \times \Spec \K \to S^\circ \times \Spec \K.  
    \end{equation}
     Under this identification, the action of $\big( \frac{\partial}{\partial z_i}\big)^p$
     is the multiplication by 
    \begin{equation}
    \label{eq:p-curvature_of_L}
    (\bar{h}^p-\bar{h}) \Big(P^{-1}\frac{\partial P}{\partial z_i} \Big)^p =\frac{\bar{h}^p-\bar{h}}{(z_i -x)^p} \in \L; 
\end{equation}
  see  Katz's formula \eqref{eq:p-curvatureoflinebundle_expl}.  
In Lemma \ref{lem:irred_crit_locus} below, we prove that $\L$ is a field and in Lemma \ref{lem:beta_anzatz},
we  prove that $\L$  is generated by $\K(z)$ and elements of the form \eqref{eq:p-curvature_of_L} for $1\leq i\leq n$.
 This will imply the theorem. 
\begin{lem}\label{lem:irred_crit_locus}
\label{lem G}
 The scheme  $\text{Crit}_P^\circ \otimes \K$
    is irreducible. In particular, $\L$ is a field (of degree $n-1$ over $\K(z)$).
\end{lem}
\begin{proof} The proof below is an adaptation of a topological argument (for $\K=\bC$) communicated to us by Andrey Gabrielov. 
     Consider the $n$-dimensional affine space $\bA^n_{\K}$ parameterizing all monic polynomials $P(x)=
       b_0 + \dots  + b_{n-1} x^{n-1} + x^n$ of degree $n$. Let $ \overline{S}_\K$ be another copy of the same affine space
       and let $\pi: \overline{S}_\K \to \bA^n_{\K}$ be the map sending $(z_1, \dots  , z_n)$ to the polynomial $P(x)= \prod_i(x-z_i)$.
       This is a degree $n!$ cover ramified over an irreducible subscheme $\text{Disc}_P \subset \bA^n_{\K}$. Consider the 
       space $\rho: \Yb_{\K} \to \bA^n_{\K}$
       parameterizing monic  polynomials $P(x)$ together with a root $c$ of its derivative $P'(x)$. Explicitly, $\Yb_{\K}$ 
       is a closed subscheme of $\bA^{n+1}_{\K}$
       given by the equation $b_1 + 2 b_2 c  + \dots  + n c^{n-1}=0$. In particular, $\Yb_{\K}$
        is isomorphic to $\bA^n_{\K}$. Denote by $\text{Disc}_{P'}$ the ramification locus of 
       $\rho$. 
       For an irreducible variety $W$ we shall denote by $\K(W)$ its field of rational functions. 
       Consider the field extensions $\K(\bA^n_{\K}) \subset \K(\overline{S}_\K ), \K(\Yb_{\K})$ determined by $\pi$ and $\rho$. We have $\L= \K(\overline{S}_\K )\otimes _{\K(\bA^n_{\K})} \K(\Yb_{\K})$. Let us prove that $\L$ is a field. To this end, choose embeddings of $ \K(\overline{S}_\K ), \K(\Yb_{\K})$ an into an algebraic closure $\overline{\K(\bA^n_{\K})}$ of  
       $\K(\bA^n_{\K})$. Since $\K(\bA^n_{\K}) \subset \K(\overline{S}_\K )$ is a Galois extension it suffices to check that the intersection $\E := \K(\overline{S}_\K )\cap \K(\Yb_{\K})$ (inside $\overline{\K(\bA^n_{\K})}$) equals $\K(\bA^n_{\K})$. Look at the ramification locus of the extension $\K(\bA^n_{\K}) \subset \E$ regarded as a closed subscheme of $\bA^n_{\K}$.  Since the extension $\K(\bA^n_{\K}) \subset \K(\overline{S}_\K )$ is unramified outside of $\text{Disc}_P$ and the extension $\K(\bA^n_{\K}) \subset \K(\Yb_{\K})$ is unramified outside of $\text{Disc}_{P'}$ the ramification locus of $\E$ is contained in
       $\text{Disc}_P\cap \text{Disc}_{P'}$. Observe that $\text{Disc}_P$ is irreducible and is not contained in  $\text{Disc}_{P'}$ (for example, the polynomial $x^n +x^2$ has a double root but 
       its derivative has only simple roots). Therefore, the intersection $\text{Disc}_P\cap \text{Disc}_{P'}$ has codimension at least $2$ in $\bA^n_{\K}$. Using 
       the purity of the branch locus \cite{nagata},  we conclude that 
       $\K(\bA^n_{\K}) \subset \E$ extends to a finite  \'etale cover of  $\bA^n_{\K}$. Since the degree of the cover is at most $n-1<p$ the cover is tamely ramified and, hence, split.
       This proves that $\L$ is a field. For the first assertion of the lemma, we
        observe that since  $\text{Crit}_P^\circ \otimes \K$ is cut by a single equation 
       inside $\bA^{n+1}_{\K}$ its every irreducible component must be of dimension $n$ and, hence, dominant over $\bA^n_{\K}$.
       The irreducibility of $\text{Crit}_P^\circ \otimes \K$ follows.         
   \end{proof} 
   \begin{rem}
       The assertion of Lemma \ref{lem:irred_crit_locus} that $\L$ is field is equivalent to the following concrete statement: the polynomial $$\frac{\partial \prod(x-z_i)}{\partial x}= nx^{n-1} - (n-1) \Big (\sum_i z_i \Big)x^{n-2} + \dots  $$
       is irreducible over $\K(z)$.
       \end{rem}
   \begin{lem}
   \label{lem:beta_anzatz}
   The field extension $\L\supset \K(z)$ is generated by the elements  $\frac{1}{(z_i -x)^p}$, $1\leq i\leq n$.
   \end{lem}
  \begin{proof}
      Since the degree of $\cL$ over $\K(z)$ is $n-1<p$ the extension is separable. Thus, it is enough to check that $\L$ is generated by $\frac{1}{z_i -x}$, $1\leq i\leq n$. We shall
      do this by showing that the map $\text{Crit}_P^\circ \otimes \K \to \bA^n_{\K} \times S^\circ $, defined by functions $\frac{1}{z_i -x}$, $1\leq i\leq n$,  and \eqref{eq:etale_map},
     is a closed embedding over the generic point of  $ S^\circ \otimes \K$.
    Since the map $\text{Crit}_P^\circ \to S^\circ$ is finite \'etale, it suffices to construct a field $\E\supset \K$ and a point $a\in S^\circ (\E)$ such that the fiber $\nu^{-1}(a)$  of  $\nu: \text{Crit}_P^\circ \to S^{\circ}$ over this point splits into $n-1$  points
      $x_j \in  \text{Crit}_P^\circ(\E)$, $1\leq j \leq n-1$, and functions $\frac{1}{z_i -x}$, ($1\leq i\leq n$), separate these points. We take for $\E$ the field $\K((s))$ of Laurent series in one variable and for $a$ the point with coordinates 
      $z_i(a)= s^i$. Then the fiber $\nu^{-1}(a)$  is given in $\bA^1_\E$ by the equation $\frac{\partial P}{\partial x}=0$, where
      $P(x)= \prod_i(x- s^i)$. 
      Rewriting the equation in the form 
      $$\sum_i \frac{1}{x-s^i}=0,$$
      we find that this equation has $n-1$ solutions in $\E$ of the form
      $$x_j = \frac{n-j}{n-j+1} s^j + O(s^{k+1}), \quad   1\leq j \leq n-1.$$
      Then, for $j>j'$,  we  have $$\frac{1}{z_j -x_j} =  \frac{1}{s^j -x_j} =  (n-j +1) s^{-j} + O(s^{1-j})$$  
      and $$\frac{1}{z_j -x_{j'}} = \frac{1}{s^j -x_{j'}}=  \frac{n-j'+1}{n-j'} s^{-j'} + O(s^{1-j'}).$$ In particular, $\frac{1}{z_j -x}$ separates 
      $x_j$ and $x_{j'}$.
       \end{proof}
   
    \end{proof}

   \subsection{Specialization of $E$ to $h\in \bF_p$}
   \label{sec 3.6}
   
   Let $p$ be a prime number that does not divide $n$, let $\tilde h$ be an integer with
   $1\leq \tilde h \leq p-1$. Consider the homomorphism  $\bZ[h, h^{-1}] \to \bF_p$ sending $h$ to
   the class $\bar h$ of $\tilde h$ in $\bF_p$ and  let
   $$E_{\bar h}:= E \otimes_{\bZ[h]} \bF_p.$$
Recall from Section \ref{ss:basechange} that the local system $E_{ \bar h}$ on $S_{\bF_p}$
can be interpreted as the relative logarithmic de Rham cohomology 
$H^1_{dR, \log}(\bP^1_S/S, \cP^{\bar h}), $ where $\cP^{\bar h}$ is the specialization of $\cP^h$, that is, the derived global sections of the complex of sheaves on $\bP^1 \times S_{\bF_p}$:
 \begin{equation}\label{eq:deRhamsp}
 \cO_{\bP^1 \times S_{\bF_p}} \rar{\nabla^{\cP^{\bar h}}}  \Omega^1_{\bP^1 \times S_{\bF_p} /S_{\bF_p} }(\log T_{\bF_p}).
\end{equation}
Denote by $\cH_{dR, \log}^i= \cH_{dR, \log}^i(\bP^1_S/S, \cP^{\bar h})$, $i=0, 1$,  the cohomology sheaves of (\ref{eq:deRhamsp}).
Then the canonical filtration on the  de Rham complex (also referred to as the conjugate filtration) gives an exact sequence of modules with connections
\begin{equation}\label{eq:conjugatetwistedlog}
  0 \to H^1(\bP^1_S, \cH_{dR, \log}^0)  \xrightarrow{\phi_{\tilde h}} E_{ \bar h}  \xrightarrow{\psi_{\tilde h}} H^0(\bP^1_S, \cH_{dR, \log}^1) \to 0.
\end{equation}
The surjectivity of the last map follows from the vanishing of $ H^2(\bP^1_S, \cH_{dR, \log}^0) $
which follows because the sheaf $\cH_{dR, \log}^0$ is quasi-coherent and $S$ is affine. 

We will prove that the boundary terms in sequence (\ref{eq:conjugatetwistedlog}) are trivial local systems, { i.e.}, isomorphic to  direct sums of several copies of $(\cO_{S_{\bF_p}}, d)$. 
To do this we shall describe the sheaves  $\cH_{dR, \log}^i$ explicitly using the generalized Cartier transform \eqref{ogusmain}. Since the $p$-curvature of $\cP^{\bar h}$  is equal to $0$ by (\ref{eq:p-curvatureoflinebundle}), we have 
\begin{equation}\label{eq:logoguscartierP}
   \cH_{dR, \log}^i \iso \Big(F^*_{\text{abs}}\Omega ^i_{\bP^1 \times S_{\bF_p}/S_{\bF_p}}\big(\, \log T_{\bF_p} \,\big)\Big)^{\nabla^{\cP^{\bar h}}=0}. 
\end{equation}
The expression at the right-hand side of the above formula stands for the sheaf of sections 
  annihilated by $\nabla^{\cP^{\bar h}}_{\!\!\frac{\partial}{\partial x}}$.

Write the absolute Frobenius $F^*_{\text{abs}}$ as the composition of two morphisms
$$
\bP^1_{\bF_p} \times S_{\bF_p} \
  \xrightarrow{F_{\bP^1_{\bF_p}}\,:=\, F_{\text{abs}} \times \Id} 
  \bP^1_{\bF_p} \times S_{\bF_p} 
  \ \xrightarrow{F_{S_{\bF_p}}\,:= \, \Id \times F_{\text{abs}} }\bP^1_{\bF_p} \times S_{\bF_p}.
  $$
Recall that both morphisms $F_{\bP^1_{\bF_p}}$ and $F_{S_{\bF_p}}$ are homeomorphisms on the underlying topological spaces (in particular, the corresponding direct image functors on the categories of sheaves are exact). Note that the differential in (\ref{eq:deRhamsp}) becomes $\cO$-linear after applying  $F_{\bP^1_{\bF_p} *}$. In particular, 
the sheaves $\cH_{dR, \log}^i$ are $\cO_{\bP^1 \times S_{\bF_p}}$-modules. The next lemma shows that they are invertible $\cO$-modules and identifies them explicitly. 
\begin{lem}\label{lem:localcohomologycomp}
 We have isomorphisms of modules with connections
 \begin{gather*}\label{eq:compofflatsections1}
F_{S_{\bF_p}}^* \Big( 
\cO_{\bP^1\times S_{\bF_p}}\Big(\Big[\frac{n\tilde h}{p} \Big] T_\infty - \sum_{i=1}^n T_i\Big) \Big) \iso   (\cO_{\bP^1 \times S_{\bF_p}} )^{\nabla^{\cP^{\bar h}}=0},\\
\label{eq:compofflatsections2}
f\in \cO_{\bP^1\times S_{\bF_p}}\Big(\Big[\frac{n\tilde h}{p} \Big] T_\infty - \sum_{i=1}^n T_i\Big) \mapsto P^{-\tilde h}  f^p,\\
F_{S_{\bF_p}}^* \Big(\Omega ^1_{\bP^1 \times S_{\bF_p}/S_{\bF_p}}
\Big( \Big (\Big[\frac{n\tilde h}{p} \Big]+1 \Big) T_\infty \Big)\Big)
\iso   
 \left(F^*_{\text{abs}}\, \Omega ^1_{\bP^1 \times S_{\bF_p}/S_{\bF_p}}
 \left(\, \log T_{\bF_p}\,\right)\right)^{\nabla^{\cP^{\bar h}}=0},\\
\label{eq:compofflatsections4}
\omega \in \Omega ^1_{\bP^1 \times S_{\bF_p}/S_{\bF_p}}
\Big( \Big (\Big[\frac{n\tilde h}{p} \Big]+1 \Big) T_\infty \Big)
 \mapsto  P^{-\tilde h} \otimes \omega.
\end{gather*}
\end{lem}
\begin{proof}
  Let us just explain the first isomorphism. If $g \in  \big(\cO_{\bP^1 \times S_{\bF_p}}\big)^{\nabla^{\cP^{\bar h}}=0}$  then $(P^{\tilde h} g)'_x =0$. Thus, $P^{\tilde h} g = F^*_{\bP^1_{\bF_p }}(f)$
  for some rational function $f$. Applying  $F_{S_{\bF_p}}^*$ to the above identity, we find that 
  $$F_{S_{\bF_p}}^* (P^{\tilde h} g) = f^p.$$
  Since $g$ is regular, we have 
  $$p \Div f \geq \tilde h \Div (F_{S_{\bF_p}}^*(P)).$$
  Thus, $f$ is a section of $F_{S_{\bF_p}}^* \Big( 
\cO_{\bP^1\times S_{\bF_p}}\Big(\Big[\frac{n\tilde h}{p} \Big] T_\infty - \sum_{i=1}^n T_i\Big) \Big)$.
  \end{proof}
  As a consequence  of the second isomorphism in Lemma \ref{lem:localcohomologycomp} combined with \eqref{eq:logoguscartierP}, we conclude that   
  $H^0(\bP^1_S, \cH_{dR, \log}^1)$ is a trivial flat bundle with a flat basis given by 
differential forms 
\begin{equation}\label{eq:basis_of_hol_forms1}
\mu_l =x^{l-1} dx \in H^0\Big(\bP^1_{S},  \Omega ^1_{\bP^1 \times S_{\bF_p}/S_{\bF_p}}
\Big( \Big (\Big[\frac{n\tilde h}{p} \Big]+1 \Big) T_\infty \Big)\Big), \qquad  1\leq l\leq \left[\frac{n\tilde h}{p} \right],
\end{equation}
\begin{equation}\label{eq:trivialization_of_loc1}
H^0(\bP^1_S, \cH_{dR, \log}^1) \iso \bigoplus_{l=1}^{[\frac{n\tilde h}{p} ]}
H^0 (S, \cO_{S_{\bF_p}}) \, \, \mu_l.
\end{equation}
Similarly, we construct a flat basis for $H^1(\bP^1_S, \cH_{dR, \log}^0)$: 
using the Serre duality  the space 
\begin{equation}\label{eq:computation_of_boundary_terms}
    H^1\Big( \bP^1_S,
\cO_{\bP^1\times S_{\bF_p}}\Big(\Big[\frac{n\tilde h}{p} \Big] T_\infty - \sum_{i=1}^n T_i\Big) \Big)
\end{equation}
 is dual to
$$
H^0\left(\bP^1_S, \;  \Omega ^1_{\bP^1 \times S_{\bF_p}/S_{\bF_p}}
 \left( \sum_{i=1}^n T_i -\left[\frac{n\tilde h}{p} \right] T_\infty \right) \right).
 $$ 
 The latter is a free module of rank $n-1 - \big[\frac{n\tilde h}{p} \big]$ generated by 
 $$
 \nu_l = x^{l-1} \frac{dx}{P}, \quad 1\leq l\leq n-1 - \left[\frac{n\tilde h}{p} \right].
 $$
Letting $\nu_l^*$ be the dual basis for  \eqref{eq:computation_of_boundary_terms}
 we obtain a flat isomorphism  
  \begin{equation}
 \label{eq:trivialization_of_loc2}
H^1(\bP^1_S, \cH_{dR, \log}^0) \iso \bigoplus_{l=1}^{n-1 -[\frac{n\tilde h}{p} ]}
H^0 (S, \cO_{S_{\bF_p}}) \, \, \nu_l^*.
\end{equation}
Recall from Section \ref{ss:solutions_in_finite} $p$-hypergeometric functions  $Q^{(pl-1)}_j (z) = Q^{(pl-1)}_j (z_1, \dots, z_n) \in H^0(S, \cO_{S_{\bF_p}})$, $1\leq l\leq \big[\frac{n\tilde h}{p} \big]$, $1\leq j\leq n$, defined by the formula
 \begin{equation}\label{eq:p-hypergeometricsol}
- \Big(\frac{P^{\tilde h}} {x-z_j}\Big)^{(p-1)}_x= \sum_{l=1}^{[\frac{n\tilde h}{p} ]} Q^{(pl-1)}_j (z) x^{p(l-1)}.
\end{equation}
Here $g^{(m)}_x$ stands for the $m$-th derivative of $g$ with respect to the $x$-variable.
Let $A_{\tilde h}$ be the $\big[\frac{n\tilde h}{p} \big] \times n$ matrix whose entries are 
 $Q^{(pl-1)}_j (z)$. Recall that $E_{\bar h}$, viewed as a $H^0 (S, \cO_{S_{\bF_p}})$-module is the quotient of the free module on
 $$\eta_i = \frac{d(x-z_j)}{x-z_j}, \quad 1\leq j\leq n,$$
 by the submodule generated by $\eta =\sum \eta_i$. Since
 $$\sum_j \Big(\frac{P^{\tilde h}} {x-z_j}\Big)^{(p-1)}_x = \frac{1}{\bar h} \Big(P^{\tilde h -1}\Big)^{(p)}_x =0, $$
 matrix $A_{\tilde h}$ defines a homomorphism of $\Gamma (Z, \cO_Z)\otimes \bF_p$-modules:
 \begin{equation}\label{eq:56}
A_{\tilde h}\colon E_{\bar h} \to \bigoplus_{l=1}^{[\frac{n\tilde h}{p} ]}  
H^0 (S, \cO_{S_{\bF_p}})\, \, \mu_l, \quad \eta_j \mapsto \sum_{l=1}^{[\frac{n\tilde h}{p} ]} Q^{(pl-1)}_j (z) \mu_l.
 \end{equation}
 We shall also consider the homomorphism
 \begin{equation}\label{eq:57}
 A^t_{p-\tilde h}\colon\bigoplus_{l=1}^{n-1-[\frac{n\tilde h}{p} ]}  
H^0 (S, \cO_{S_{\bF_p}}) \, \, \nu_l^* \to E_{\bar h}.  
\end{equation}
 given by  the $n\times (n-1-\big[\frac{n\tilde h}{p} \big])$  matrix $A_{p-\tilde h}^t$. The following result (applied to $(\cV, \nabla^{\on{KZ}, h}) =  E_{-\bar h}$) implies Theorem \ref{thm:main3_introduction} from the Introduction.
\begin{thm}
\label{th:main_result_for_rational_h}

${}$

\begin{itemize}

\item[(i)]

Under the isomorphisms \eqref{eq:trivialization_of_loc1}, \eqref{eq:trivialization_of_loc2} the sequence \eqref{eq:conjugatetwistedlog} 
of flat bundles has the form
\begin{equation}\label{eq:conjugatetwistedlogbis}
  0 \to  \bigoplus_{l=1}^{n-1 -[\frac{n\tilde h}{p} ]}
H^0 (S, \cO_{S_{\bF_p}}) \, \, \nu_l^* \
\xrightarrow{-h A_{p-\tilde h}^t}\
 E_{\bar h} \xrightarrow{A_{\tilde h}} \
 \bigoplus_{l=1}^{[\frac{n\tilde h}{p} ]}
H^0 (S, \cO_{S_{\bF_p}}) \, \, \mu_l \to 0,
\end{equation}
where each $\nu_l^*$ and  $\mu_l$ is flat.
In particular, we have
\begin{equation}
  A_{\tilde h}  A_{p-\tilde h}^t=0.
\end{equation}
\item[(ii)] 
If  $\big[\frac{n\tilde h}{p} \big]  \ne n-1$ then morphism $A^t_{p-\tilde h}$ induces an isomorphism
\begin{equation}\label{eq:fltat_sections_of_E}
\bigoplus_{l=1}^{n-1-[\frac{n\tilde h}{p} ]} H^0 (S, \cO_{S_{\bF_p}}^p) \, \,\nu_l^* \iso E_{\bar h}^{\nabla =0}.
\end{equation}
 \item[(iii)] The duality  $E_{ \bar h} \otimes E_{- \bar h} \to  \cO_{S_{\bF_p}}$ (see Corollary \ref{cor:duality_section2}) restricts to the zero map 
 on $ H^1(\bP^1_S, \cH_{dR, \log}^0(\bP^1_S/S, \cP^{\bar h})) \otimes  H^1(\bP^1_S, \cH_{dR, \log}^0(\bP^1_S/S, \cP^{-\bar h}))$
 and induces a perfect pairing 
 \begin{equation}\label{eq:lagrangian_property}
     H^1(\bP^1_S, \cH_{dR, \log}^0(\bP^1_S/S, \cP^{\bar h})) \otimes  H^0(\bP^1_S, \cH_{dR, \log}^1(\bP^1_S/S, \cP^{-\bar h}))\to \cO_{S_{\bF_p}}
 \end{equation}
 \end{itemize}
\end{thm}
\begin{proof} Part (i). First,  observe that  the image of $\phi_{\tilde h}$ in \eqref{eq:conjugatetwistedlog} is precisely the kernel 
the restriction map 
\begin{equation}\label{eq:restrictionandconjfilt}
E_{\bar h}=H^1_{dR, \log}(\bP^1_S/S, \cP^{\bar h})\to  H^1_{dR}(U/S, \cP^{\bar h}), \quad U=\bP^1_S\,\setminus\, T.
\end{equation} 
This follows from the vanishing 
$H^1(U, \cH_{dR, \log}^0|_{U})=0$ and the injectivity of the restriction map $H^0(\bP^1_S, \cH_{dR, \log}^1) \to H^0(U, \cH_{dR, \log}^1|_{U})$.

Let us check that under trivialization  (\ref{eq:trivialization_of_loc1}) morphism $\psi_{\tilde h}$  is given by matrix $ A_{\tilde h}$. This amounts to verifying that, for every
$1\leq j\leq n$, we have that 
\begin{equation}\label{eq:-1}
    C(\eta_j)= - P^{-\tilde h} \Big(\frac{P^{\tilde h}} {x-z_j}\Big)^{(p-1)}_x \otimes dx,
\end{equation}
where $C$ is the generalized Cartier operator for $\cP^{\bar h}$. 
Since $ \cH_{dR, \log}^1$ is locally free it suffices to check that $\psi_{\tilde h}=A_{\tilde h}$ on $U_{\bF_p} \subset \bP^1_{S_{\bF_p}}$.
But the local system $\cP^{\bar h}$ restricted to $U_{\bF_p}$ is trivial:
 $$\cP^{\bar h}|_{U_{\bF_p}} \iso \cO_{U_{\bF_p}}, \quad f\mapsto P^{\tilde h} f.$$ 
Formula (\ref{eq:-1}) is now a consequence of (\ref{ogus_explicit_formula}) applied to the trivial local system.

Let us prove that under trivialization  (\ref{eq:trivialization_of_loc2}) morphism $\phi_{\tilde h}$  is given by matrix $-h A_{p- \tilde h}^t$. To do this we shall use the duality $E_{- \bar h}^* \iso  E_{\bar h}$ (see Corollary \ref{cor:duality_section2}).
Consider the sequence (cf. \eqref{eq:conjugatetwistedlog}) 
\begin{equation}\label{eq:conjugatetwistedlogdual}
  0 \to H^1(\bP^1_S, \cH_{dR, \log}^0( \cP^{-\bar h}))  \xrightarrow{\phi_{p- \tilde h}} E_{- \bar h}  \xrightarrow{\psi_{p- \tilde h}} H^0(\bP^1_S, \cH_{dR, \log}^1 ( \cP^{-\bar h})) \to 0.
\end{equation}
Using Lemma \ref{lem:localcohomologycomp} we find $\cH_{dR, \log}^0( \cP^{\bar h}) \iso \Hom_{\cO_{\bP^1 \times S_{\bF_p}}}(\cH_{dR, \log}^1( \cP^{-\bar h}),
\Omega^1_ {\bP^1 \times S_{\bF_p}/S_{\bF_p}})$. We have 
\begin{equation}\label{eq:lagrangian_proof}
 (\phi_{\tilde h} v, w)= (v,  \psi_{p- \tilde h} w), \quad \; v\in H^1(\bP^1_S, \cH_{dR, \log}^0( \cP^{\bar h})), \; w\in  E_{- \bar h},
\end{equation}
where the left-hand side is given by the pairing $E_{\bar h} \otimes E_{- \bar h}\to \bP^1_{S_{\bF_p}}$ and the right-hand side by the Serre duality.
The formula $\phi_{\tilde h}= -h A_{p- \tilde h}^t$ follows from $ \psi_{p- \tilde h} = A_{p- \tilde h}$ proven above and Corollary \ref{cor:duality_section2}.
The completes the proof of part (i). We note that \eqref{eq:lagrangian_proof} also proves  part (iii).

 Part (ii). Applying   the
left exact functor $M\mapsto M^{\nabla=0}$ to  sequence \eqref{eq:conjugatetwistedlog}, we find
$$ 0 \to H^1(\bP^1_S, \cH_{dR, \log}^0)^{\nabla=0}  \xrightarrow{\phi_{\tilde h}} E_{\bar h}^{\nabla^E=0}  \xrightarrow{\psi_{\tilde h}} H^0(\bP^1_S, \cH_{dR, \log}^1)^{\nabla=0}.$$
 We have to show that the last morphism is equal to $0$.  
 Since $H^0_{dR, \log}(\bP^1_S/S, \cP^{\bar h})=0$
 the canonical morphism $$H^1_{dR, \log}(\bP^1_S, \cP^{\bar h}) \to H^1_{dR, \log}(\bP^1_S/S, \cP^{\bar h})^{\nabla=0}=E_{\bar h}^{\nabla^E=0} $$
 is an isomorphism. We have a commutative diagram
 \begin{equation}\label{dia:relative_v_absolute}
\def\normalbaselines{\baselineskip20pt
\lineskip3pt  \lineskiplimit3pt}
\def\mapright#1{\smash{
\mathop{\to}\limits^{#1}}}
\def\mapdown#1{\Big\downarrow\rlap
{$\vcenter{\hbox{$\scriptstyle#1$}}$}}
\begin{matrix}
H^1_{dR, \log}(\bP^1_S, \cP^{\bar h}) & \xrightarrow{}  &  H^0(\bP^1_S, \cH_{dR, \log}^1(\bP^1_S, \cP^{\bar h} ))  \cr
 \mapdown{}  &  &\mapdown{\alpha}  \cr
H^1_{dR, \log}(\bP^1_S/S, \cP^{\bar h}) & \xrightarrow{}  &  H^0(\bP^1_S, \cH_{dR, \log}^1(\bP^1_S/S, \cP^{\bar h} )).
\end{matrix}
\end{equation}
We will complete the proof by showing that morphism $\alpha$ in diagram (\ref{dia:relative_v_absolute}) is equal to $0$.

Using  the generalized Cartier transform, we have that
$$\cH_{dR, \log}^1(\bP^1_S, \cP^{\bar h} )\iso \Big(F^*_{\text{abs}}\Omega ^i_{\bP^1 \times S_{\bF_p}}\big(\, \log T_{\bF_p}\,\big)\Big)^{\nabla^{\cP^{\bar h}}=0}.$$
The expression at the right-hand side of the above formula stands for the sheaf of sections 
  annihilated by all germs of vector fields on $\bP^1_{S_{\bF_p}}$. This sheaf can be described
  explicitly:
  \begin{gather*}
 \Omega ^1_{\bP^1 \times S_{\bF_p}}\big(\, \log T_{\bF_p} \,\big) \otimes \cO_{\bP^1 \times S_{\bF_p}}\Big(\Big[\frac{n\tilde h}{p} \Big] T_\infty - \sum_{i=1}^n T_i\Big)\\
\iso F_{\text{abs}*}\Big(F^*_{\text{abs}}\Omega ^1_{\bP^1 \times S_{\bF_p}}\big(\, \log T_{\bF_p}\,\big)\Big)^{\nabla^{\cP^{\bar h}}=0},\\
\omega 
 \mapsto  P^{-\tilde h} \otimes \omega.
\end{gather*}
The desired  vanishing of $\alpha$  follows from the next result.
\begin{lem}
    Let $K$ be a field and $a=(a_1, \dots, a_n)$ be a $K$-valued point of $S$.  Consider the 
  exact sequence of sheaves
\begin{equation}\label{eq:exactsequencelogks}
    0\to \bigoplus_{i=1}^n \cO_{\bP^1_S} dz_i \to \Omega ^1_{\bP^1_S}(\log T) \to 
    \Omega ^1_{\bP^1_S/S}(\log T) \to 0,
\end{equation}
 and let    
    \begin{equation}\label{eq:exactsequencelogksfiber}
    0\to \bigoplus_{i=1}^n \cO_{\bP^1_a} dz_i \to (\Omega ^1_{\bP^1_S}(\log T))|_{\bP^1_a} \rar{\gamma} 
    \Omega ^1_{\bP^1_a}(\log T) \to 0
    \end{equation}
    be the restriction of (\ref{eq:exactsequencelogks}) to the fiber $\bP^1_a$ of the projection $\bP_S \to S$ over $a$. Fix an integer $m$ and denote by $D$ the divisor $m T_\infty - \sum_{i=1}^n T_i$ restricted to the projective line $\bP^1_a$.
    Then, if $m< n-1$, the morphism 
    $$H^0(\bP^1_a, (\Omega ^1_{\bP^1_S}(\log T))|_{\bP^1_a} \otimes  \cO_{\bP^1_a}(D))\to 
  H^0(\bP^1_a,  \Omega ^1_{\bP^1_a}(\log T) \otimes  \cO_{\bP^1_a}(D) ) $$
  induced by $\gamma$ is equal to $0$.
\end{lem}

\begin{proof}
Consider the tensor product of  sequence (\ref{eq:exactsequencelogksfiber}) with 
$\cO_{\bP^1_a}(D)$ and the associated long exact sequence of cohomology groups. We need to show that the connecting homomorphism 
$$  H^0(\bP^1_a,  \Omega ^1_{\bP^1_a}(\log T) \otimes  \cO_{\bP^1_a}(D) ) \rar{\delta}
\bigoplus_{i=1}^n H^1(\bP^1_a,\cO_{\bP^1_a}(D)) dz_i
$$
is injective. 
We give an explicit formula for $\delta$. Identify 
$$H^1(\bP^1_a,\cO_{\bP^1_a}(D))\iso \coker(H^0(\bP^1_a,\cO_{\bP^1_a}(m[\infty])  ) \to 
H^0(\bP^1_a,\cO_{\bP^1_a}(m[\infty])/\cO_{\bP^1_a}(D)  ))$$
with cokernel of the restriction map
 \begin{equation}\label{eq:rest_and_coh_comp}
\text{Pol}^{\leq m} \to \text{Fun}(\{a_1, \dots, a_n\}),
\end{equation}
where $\text{Pol}^{\leq m}\subset K[x]$ is the space of polynomials of degree $\leq m$ and 
$\text{Fun}(\{a_1, \dots, a_n\})$ is the space of $K$-valued functions on the set $\{a_1, \dots, a_n\}$.
We also identify the space of global sections of the invertible sheaf 
$ \Omega ^1_{\bP^1_a}(\log T) \otimes  \cO_{\bP^1_a}(D)  =  \Omega ^1_{\bP^1_a}\otimes  \cO_{\bP^1_a}(m[\infty])$ with $\text{Pol}^{\leq m-2}dx$.
Note that the latter space is $0$ if $m<2$ (and, thus, the conclusion of the Lemma trivially holds). 
Hence, we shall assume, for the remaining part of the proof that $m\geq 2$. 
For every  $f(x)dx\in \text{Pol}^{\leq m-2}dx$ the formula
$$f(x)dx = f(x)(x-z_i) \frac{d(x-z_i)}{x-z_i} + f(x)dz_i$$
shows that the coefficients $c_i$ in
$$\delta(f(x)dx)= \sum_{i=1}^{n}c_i dz_i \in \bigoplus_{i=1}^n H^1(\bP^1_a,\cO_{\bP^1_a}(D)) dz_i$$ are, under our identification,  given by
$$c_i(j)= \delta_{ij} f(a_i), \quad 1\leq i,j \leq n.$$
Note that,  for every non-zero $f\in \text{Pol}^{\leq m-2}$ there exists $i$ such that $f(a_i)\ne 0$. But then the class of the function $c_i$ in the cokernel of the restriction map 
(\ref{eq:rest_and_coh_comp}) is also non-zero: the function $c_i$ has precisely $n-1$ zeroes so it cannot be a restriction of a polynomial of degree $\leq m < n-1$.
\end{proof}
This completes the proof of part (ii). Part (iii) has been already proven.
\end{proof}

\section{Cohomology of a certain family of curves}
\label{sec 4} For the duration of this section we fix coprime integers $q>0$, $n>1$, and set $\Lambda = \bZ[\frac{1}{nq}]\subset \bQ$, $S= \Spec \Lambda [z_i, (z_i-z_j)^{-1}, 1\leq i< j \leq n]$. We shall construct a smooth projective curve $X$ over $S$ equipped with an action the group scheme $\mu_q$ of $q$-th roots of unity. We show  that
the isotypic components of $H^1_{dR}(X/S)$ are closely related to the specialization of
the local system $E$ to $h\in \bQ$. We prove a certain non-degeneracy property of the Kodaira-Spencer map associated with 
the Hodge filtration on $H^1_{dR}(X/S)$ and give an application to the study of the $p$-curvature of the Gauss-Manin connection on $H^1_{dR}(X_{\bF_p}/S_{\bF_p})$. 

\subsection{Family of curves}\label{ss:curve} 
Let  $W\mono \bA^2 _S=  \bA^2 \times S $ be the closed subscheme  defined by equation
$$y^q= \prod_{i=1}^{n}(x -z_i).$$
The scheme $W$ is smooth over $S$. The projection $\bA^2 _S \to \bA^1_S$, $(x,y) \mapsto x$,   yields a quasi-finite morphism
\begin{equation}\label{eq:proj}
W \mono \bA^2 _S \to  \bA^1_S\mono  \bP^1_S.
\end{equation}
Let $$\rho: X \to \bP^1_S$$ be the normalization of $\bP^1_S$ in $W$ (\cite[Tag 0BAK]{stacks}). Thus, $X$ is a normal integral scheme equipped with a finite morphism $\rho$ such that $\rho^{-1}( \bP^1_S\setminus \{\infty \times S\})=
W$.
\begin{lem}\label{lm:curve} The composition $\pi: X\rar{\rho}\bP^1_S \to  S$ is proper and smooth. The morphism $\rho$ is flat. Fibers of $\pi$ are  geometrically connected fibers curves of genus 
$$g= \frac{qn -q -n +1}{2}.$$  
\end{lem}
\begin{proof}
 We shall construct $\rho: X \to \bP^1_S$ by gluing two 
affine charts $W$, $W'$.
 The one, $W$, is already defined. To describe the second one, choose positive integers $a, b$ such that
 $an - bq =-1$. Define  $W'$  to be the closed subscheme of the affine space  $\Spec \cO(S)[v,u]\iso \bA^2 _S $  given by the equation
$$u^q= v   \prod_{i=1}^{n}(1 -z_i v)^a.$$
Scheme $X$ is obtained by gluing $W$ and $W'$.
The glueing maps between the two charts are given by
$$v=x^{-1}, u = y^a x^{-b}; \quad  x= v^{-1},  y = u^{-n} \prod_{i=1}^{n}(1 -z_i v)^b.$$
Cover $\bP^1_S$ by two copies  of the affine line, $\Spec \cO(S)[x] $,  $\Spec \cO(S)[v] $, glued via $x \mapsto x^{-1}$. Then  the map $\rho: X \to \bP^1_S$ carries $W$ to $\Spec \cO(S)[x] $ by (\ref{eq:proj}) and $W'$ to 
 $\Spec \cO(S)[v] $ by $(v,u)\mapsto v$.
 Using the Jacobian criterion $X$ is smooth over $S$. It is also connected. The finite morphism $\rho: X \to \bP^1_S$ coincides with (\ref{eq:proj})  over $\bA^1_S\mono \bP^1_S$. Thus, $\rho: X \to \bP^1_S$  
is the normalization of  (\ref{eq:proj}). Flatness of $\rho$ follows from the Miracle Flatness Theorem asserting that a finite surjective morphism of connected regular schemes is flat
 (\cite[Tag 00R4]{stacks}).

\smallskip

The formula for the genus of the fiber follows from the Riemann-Hurwitz formula: if $z=(z_1, 
{\dots}, z_n)$ is a $k$-point of $S$, 
where $k$ is a field, then  the fiber $X_z$ of $\pi$ over $z$
admits a finite degree $q$ map $\rho_z: X_z\to \bP^1_k$ ramified exactly over points $z_i$, $1\leq i\leq n$, and over $\infty$. Moreover, the   ramification index $e$ over each of these points is $q$ (which is coprime to $\Char k$). Thus,
$$2-2g = 2 \deg \rho_z - (n+1)(e-1)=  2 q - (n+1)(q-1),$$
and we win.  
\end{proof}
Using the affine charts introduced in the proof of Lemma \ref{lm:curve} we define certain divisors on $X$.  
For $1\leq i \leq n$, let  $D_i \mono W \mono X$ (resp. $T_i\mono \bA^1_S \subset \bP^1_S$) be the divisor given by equations $x=z_i$, $y=0$ (resp. $x=z_i$). Also, let 
$D_\infty\mono W' \mono X$ be the divisor given by equations
 $v=0$, $u=0$. Finally, put
$T_\infty = \{\infty \} \times S \mono
\bP^1 \times S =\bP^1_S$.  
We have that
 \begin{equation}\label{eq:ramification}
\rho^*[T_i]= q [D_i], \quad \rho^*[T_\infty]= q [D_\infty].
\end{equation}

The finite group scheme $\mu_q$ of roots of unity acts on $W$ 
 by the formula $\sigma_\epsilon(x,y)= (x, \epsilon y)$, $\epsilon \in \mu_q$. 
 This action extends uniquely to an action of $\mu_q$ on $X$.
 
 \subsection{De Rham cohomology of $X$: the set up.}
 Consider the relative de Rham cohomology
$H^1_{dR}(X/S)= R^1\pi_* (\Omega_{X/S}^\bullet, d) $. This is a locally free $\cO_S$-module  of rank $2g$
 equipped with the Gauss-Manin  connection $\nabla$, the Hodge 
  filtration $F^1\subset H^1_{dR}(X/S)$,
  \begin{equation}\label{eq:Hodgefilt}
F^1\iso  \pi_* \Omega_{X/S}^1, \qquad H^1_{dR}(X/S)/F^1 \iso R^1\pi_*\cO_X.
\end{equation}
Both $F^1$ and $H^1_{dR}(X/S)/F^1$ are locally free $\cO_S$-modules of rank $g$.
  The Poincare duality induces a perfect skew symmetric bilinear form 
 \begin{equation}\label{eq:Poincare}
H^1_{dR}(X/S) \otimes _{\cO_S} H^1_{dR}(X/S)  \to \cO_S,
\end{equation}
such that $F^1$ is Lagrangian with respect to (\ref{eq:Poincare}).

The action of $\mu_q$ on  $H^1_{dR}(X/S)$ determines a decomposition of the latter into a direct sum of isotypic components numbered by characters $\Hom(\mu_q, \bG_m)= \bZ/q\bZ$ of the group scheme $\mu_q$:
\begin{equation}\label{eq:decomposition}
H^1_{dR}(X/S)= \bigoplus _{\bar r\in \bZ/q\bZ}  H^1_{dR}(X/S)_{\bar r}, \quad F^1= \bigoplus _{\bar r \in \bZ/q\bZ}  F_{\bar r}^1,
 \end{equation}
Thus, for every $\epsilon \in \mu_q$, the endomorphism $\sigma^*_\epsilon: H^1_{dR}(X/S) \to H^1_{dR}(X/S) $ preserves each summand in (\ref{eq:decomposition}) and   acts on $H^1_{dR}(X/S)_{\bar r}$ as
 $\epsilon ^{-\bar r} \Id$. The Gauss-Manin  connection $\nabla$ preserves each summand $H^1_{dR}(X/S)_{\bar r}$. The Poincar\'e duality (\ref{eq:Poincare}) induces isomorphisms
 \begin{equation}\label{eq:duality_and_mu}
 H^1_{dR}(X/S)_{\bar r} \iso H^1_{dR}(X/S)_{- \bar r}^*, \quad
F^1_{\bar r} \iso (H^1_{dR}(X/S)/F^1)_{-\bar r}^*.
\end{equation}
\subsection{De Rham cohomology of an open curve $X^{\circ}\subset X$.}
Let $X^\circ: = X\setminus (D_\infty \cup \bigcup D_i) $ be the complement to the ramification divisor of $\rho$, $U:=\bP^1_S \setminus (T_\infty \cup \bigcup T_i)$. Then the morphism $\rho: X^\circ \to U$ is a finite  \'etale cover of degree $q$. The group scheme 
$\mu_q$ acts on $X^\circ$ by deck transformations. 
It follows that $\rho_*\cO_{X^\circ}$ is equipped with an 
integrable connection 
$$\nabla: \rho_*\cO_{X^\circ} \to \rho_*\cO_{X^\circ} \otimes _{\cO_U}  \Omega^1_{U}$$
induced by the tautological connection on  $\cO_{X^\circ}$. Each isotypic component $(\rho_*\cO_{X^\circ})_{\bar r}$ is stable under $\nabla$ and has  rank $1$. 
We shall describe this local system explicitly. Pick $r\in \bZ$ congruent to $\bar r$ modulo $q$ and set $\bar h=  - \frac{r}{q}$.
Then $(\rho_*\cO_{X^\circ})_{\bar r}$  is generated by the global section $\frac{1}{y^r}$
with
$$\nabla\left(\frac{1}{y^r}\right)= \frac{\bar h}{y^r}  \eta, $$
where $\eta=\frac{dP}{P}$ is defined in \eqref{log-1-form}. This identifies $(\rho_*\cO_{X^\circ})_{\bar r}$ with the specialization $\cP^{\bar h}$ of the local system 
$\cP^{h}$ introduced in Section \ref{ss:loc_system}. Consequently, we have
\begin{equation}\label{eq:coh_of_open_curve}
 H^\bullet_{dR}(X^{\circ}/S)_{\bar r}\iso H^\bullet\big(\Gamma(U, \cO_U)\xrightarrow{\nabla^{\cP^{\bar h}}}\Gamma(U, \Omega^1_{U/S})\big), \quad \nabla^{\cP^{\bar h}}(f)= df +\bar h f \eta. 
\end{equation}

\subsection{De Rham cohomology of $X_{\bQ}$.}
For  $\bar r \ne 0$, the restriction map 
$$H^\bullet_{dR}(X/S)_{\bar r} \to H^\bullet_{dR}(X^{\circ}/S)_{\bar r}$$
induces  an isomorphism of vector spaces over $\bQ$. 
Indeed, we have an exact sequence
\linebreak
 $\cO_S$-modules
$$ 0\to H^1_{dR}(X/S)\otimes \bQ    \rar{} H^1_{dR}(X^{\circ}/S)\otimes \bQ  \rar{\text{Res}} \cO_S^{\oplus n+1}\otimes \bQ \rar{\sum}   \cO_S \otimes \bQ  \to 0,$$
where $\text{Res}$ stands for the residue map at the divisors $D_1, {\dots}, D_n, D_\infty$.  The invariance of $\text{Res}$ under the $\mu_q$-action forces the first map to be an isomorphism on the  isotypic components with $\bar r \ne 0$.
Also, for $\bar r \ne 0$,  the  restriction morphism 
$$H^\bullet_{dR, \log}(\bP^1_S/S, \cP^{\bar h})  \to H^\bullet_{dR}(U/S,  \cP^{\bar h})$$
is an isomorphism  after tensoring with $\bQ$. This follows from Proposition \ref{pr:log_to_open}.
Summarizing, we conclude that, for $\bar r \ne 0$,   $\eta_i =\frac{d(x-z_i)}{x-z_i}, \bar h= -\frac{r}{q}, 1\leq i \leq n$, we have
\begin{equation}\label{eq:isooverrationals}
H^1_{dR}(X/S)_{\bar r}\otimes \bQ \iso  E_{\bar h}\otimes \bQ, \quad 
 \left[\frac{\eta_i}{y^r}\right]\mapsto  [\eta_i]\,, 
\end{equation}
\begin{equation}\label{eq:isooverrationalsbis}
H^1_{dR}(X/S)_{- \bar r}\otimes \bQ \iso  E_{- \bar h}\otimes \bQ, \quad 
 \left[y^r \eta_i \right]\mapsto  [\eta_i]\,. 
\end{equation}
We remark that the isomorphisms \eqref{eq:isooverrationals},  \eqref{eq:isooverrationalsbis} intertwine the Poincar\'e duality \eqref{eq:duality_and_mu} and the duality   $E_{- \bar h}^*\otimes \bQ \iso 
E_{\bar h}\otimes \bQ$ defined in Section \ref{ss:duality}. In particular, using Corollary \ref{cor:duality_section2}, we have 
 \begin{equation}\label{eq:poincare_duality_expl}
  ([ y^{-r}\eta_i], [y^r \eta_j])= -h^{-1}(\delta_{ij} - \frac{1}{n}), \quad \;
1\leq i, j \leq n\,. 
 \end{equation}

Using Proposition \ref{pr:basis_for_E} we find that $H^1_{dR}(X/S)_{\bar r}\otimes \bQ$
is a  free module of rank $n-1$ over $\Gamma(S,  \cO_S) \otimes  \bQ$ generated by classes of rational $1$-forms
\begin{equation}\label{eq:forms}
\omega^{(r)}_i= \frac{\eta_i}{y^r}=  \frac{d(x-z_i)}{y^r(x-z_i)}, \quad 1\leq i \leq n,
\end{equation}
with a single relation: $\sum_{i=1}^n [\omega_i^{(r)}] =0$.
It follows that the same assertion holds after inverting in $\Lambda$ sufficiently many prime integers. In the next subsection we shall make this bound explicit. 
\subsection{De Rham cohomology of $X$.}
For an integer $0<r<q$,  set 
\begin{align}
\label{eq:N}
N_{r,n,q} = N= \prod_a \left(q a - nr\right),
\end{align} 
where $a$ runs over all integers in the interval $\big[0,  \frac{rn}{q}\big]$. 

\begin{thm}\label{th:basis} 
${}$
\begin{itemize}
\item[(i)]  We have that $H^1_{dR}(X/S)_{\bar 0}=0$ and, for every $\bar r \ne 0$, the $\cO_S$-module   $H^1_{dR}(X/S)_{\bar r}$ is locally free of rank  $n-1$.
\item[(ii)] Let $0<r<q$ be an integer congruent to $\bar r$ modulo $q$, $\bar h= -\frac{r}{q}$. 
Then there exist unique isomorphisms of local systems over $S\times \Spec  \bZ[N_{r,n,q}^{-1}]$
\begin{equation}\label{eq:isooverintegers}
     H^1_{dR}(X/S)_{\bar r}\otimes \bZ[N_{r,n,q}^{-1}]  \iso E_{\bar h} \otimes \bZ[N_{r,n,q}^{-1}]\,,
\end{equation}
\begin{equation}\label{eq:isooverintegersbis}
     H^1_{dR}(X/S)_{-\bar r}\otimes \bZ[N_{r,n,q}^{-1}]  \iso E_{-\bar h} \otimes \bZ[N_{r,n,q}^{-1}]\,,
\end{equation}
extending \eqref{eq:isooverrationals} and  \eqref{eq:isooverrationalsbis}\footnote{We remark that \eqref{eq:isooverintegersbis} and \eqref{eq:isooverintegers} are two different statements:
the substitution $\bar r \mapsto -\bar r$ in \eqref{eq:isooverintegers} gives $H^1_{dR}(X/S)_{- \bar r}\otimes \bZ[N_{q-r,n,q}^{-1}]  \iso E_{- \bar h -1}\otimes \bZ[N_{q-r,n,q}^{-1}]$.}.
\item[(iii)] Let $0<r<q$ be an integer congruent to $\bar r$ modulo $q$. Then,  for each $1\leq i \leq n$, there exists  a unique class in 
$H^1_{dR}(X/S)_{\bar r}\otimes \bZ[N_{r,n,q}^{-1}]$  (resp.  $H^1_{dR}(X/S)_{-\bar r}\otimes \bZ[N_{r,n,q}^{-1}]$)    whose restriction to  $H^1_{dR}(X^\circ_\bQ/S)_{\bar r}$ (resp. $H^1_{dR}(X^\circ_\bQ/S)_{-\bar r}$)
is represented by the $1$-form $\omega^{(r)}_i$ (resp. $\omega^{(-r)}_i$). The class is denoted by
$[\omega^{(r)}_i]$ (resp. $[\omega^{(-r)}_i]$). Moreover, isomorphism \eqref{eq:isooverintegers} (resp.  \eqref{eq:isooverintegersbis}) carries each $[\omega^{(r)}_i]$ (resp. $[\omega^{(-r)}_i]$)  to $[\eta_i]$.
In particular,
$H^1_{dR}(X/S)_{\bar r}\otimes \bZ[N_{r,n,q}^{-1}]$ is a free module over $\cO(S)\otimes \bZ[N_{r,n,q}^{-1}]$ 
 on generators $[\omega^{(r)}_i]$, $1\leq i \leq n$, and a single relation $\sum_{i=1}^n [\omega^{(r)}_i] =0$. Furthermore,  the Gauss-Manin connection acts on $[\omega^{(r)}_i]$ by
 formula \eqref{eq:explicit_formula_for_connection}, and the Poicar\'e pairing on $H^1_{dR}(X/S)$ satisfies \eqref{eq:poincare_duality_expl}.
 \end{itemize}
\end{thm}
\begin{proof}  Since $\rho$ is $\mu_q$-invariant,
\  $\mu_q$ acts on  $\rho_* \Omega_{X/S}^{i}$ yielding a decomposition
  \begin{equation}\label{eq:pushforwarddecomp}
 \rho_* \Omega_{X/S}^\bullet \iso \bigoplus _{\bar r \in \bZ/q\bZ}( \rho_* \Omega_{X/S}^\bullet)_{\bar r}
 \end{equation}
 compatible with de the Rham differential. Consequently, 
    we have
\begin{equation}\label{eq:begining}
  H^\bullet_{dR}(X/S)_{\bar r}\iso R^\bullet \Gamma (\bP^1_S,   (\rho_* \cO_{X})_{\bar r} \xrightarrow{d}  (\rho_* \Omega_{X/S})_{\bar r}).
\end{equation}
Note that, for each $\bar r$, $( \rho_* \Omega_{X/S}^\bullet)_{\bar r}$ is an invertible $\cO_{\bP^1_S}$-module.  
We shall describe it explicitly.  
\begin{lem}
\label{lm:directimage}
Let $0< r< q$ be an integer congruent to $\bar r$ modulo $q$.
Then, we have that
\begin{equation}\label{eq:directimage1}
\cO_{\bP^1_S}\left(\, \left[\frac{rn}{q}\right][T_\infty]  - \sum_{i=1}^n [T_i]\,\right) 
 \iso   ( \rho_* \cO_{X})_{\bar r}, \quad f\mapsto \frac{\rho^*f}{y^r}\,,
 \end{equation}
\begin{equation}\label{eq:directimage2}
\Omega^1_{\bP^1_S/S}\left(\, \left[\frac{rn+q-1}{q}\right]
[T_\infty] \,\right)  \iso   ( \rho_* \Omega_{X/S}^1)_{\bar r}, \quad \omega \mapsto \frac{\rho^*\omega }{y^r}\,.
 \end{equation}
\end{lem}
 \begin{proof} Let us explain  isomorphism (\ref{eq:directimage2}).
Using  formula (\ref{eq:ramification}) we conclude that $\rho^*:   \rho^* \Omega^1_{\bP^1_S/S} \to  \Omega_{X/S}^1$ yields an isomorphism:
$$ 
\rho^* \Omega^1_{\bP^1_S/S} \iso  \Omega_{X/S}^1
\left(- (q-1) \left([D_\infty]  + \sum_{i=1}^n [D_i]\right)\right).
$$
It follows that $\omega$ is a rational form on $\bP^1_S$ with 
$$\text{div}(\omega) = m_\infty [T_\infty] + \sum_i m_i [T_i],$$ 
for some integers $m_i$, $m_\infty$, then 
$$
\text{div}\left(\frac{\rho^* \omega}{y^r}\right) =
 (q m_\infty + q-1 + rn) [D_\infty] + \sum_i (q m_i + q-1 -r )[D_i].
 $$
Observe that  $m_\infty = -[\frac{rn+q-1}{q}]$ (resp.  $m_i = 0$)  is the smallest integer with  $q m_\infty + q-1 + rn\geq 0$ (resp. $q m_i + q-1 -r\geq 0$). 
Now we can prove (\ref{eq:directimage2}). If $\omega$ is a section of the left-hand side of  (\ref{eq:directimage2}) then $\text{div}(\frac{\rho^* \omega}{y^r}) \geq 0$.  Thus, the morphism in  (\ref{eq:directimage2}) is well-defined.  
To prove that it is an isomorphism observe that if $\omega'$ is a section of $( \rho_* \Omega_{X/S}^1)_{\bar r}$ then $y^r \omega'$ is $\mu_q$-invariant and, hence, has the form $\rho^* \omega$ for some rational form on  $\bP^1_S$.
Since $\text{div}(\frac{\rho^* \omega}{y^r})\geq 0$, it follows that $\omega$ is section of the left-hand side of  (\ref{eq:directimage2}). 

Proof of  (\ref{eq:directimage1}) is similar but easier. 
 \end{proof}
 Let us return to the proof of the theorem.

\smallskip
For part {(i)}, observe that since $H^1_{dR}(X/S)$ is a locally free $\cO_S$-module its direct summands $H^1_{dR}(X/S)_{\bar r}$ are also locally free. The formula for their ranks 
follows from \eqref{eq:isooverrationals} (for $\bar r \ne 0$) and from the formula for genus of $X$ (Lemma \ref{lm:curve}). 

For part  {(ii)}, using notation from Section
\ref{ss:calibration}, formula \eqref{eq:begining}, and Lemma \ref{lm:directimage}
we find
\begin{equation}
     \Big((\rho_* \cO_{X})_{\bar r} \xrightarrow{d}  (\rho_* \Omega_{X/S})_{\bar r}\Big) \iso  C_{m_\infty, 0, \dots ,0},
\end{equation}
where $m_\infty =\left[\frac{rn+q-1}{q}\right]$. By Lemma \ref{lm:divisage}  the map $C_{m_\infty, 0, \dots ,0} \to C_{m_\infty, 1, \dots , 1}$ is a quasi-isomorphism and
$C_{1, 1, \dots ,1} \to C_{m_\infty, 1, \dots , 1}$ is a quasi-isomorphism after inverting $N$. This proves \eqref{eq:isooverintegers}. Proof of \eqref{eq:isooverintegersbis} is similar.

 For part  {(iii)}, define $[\omega^{(r)}_i]$  to be the preimage of $[\eta_i]$ under the isomorphism \eqref{eq:isooverintegers}. Its restriction to $H^1_{dR}(X^\circ_\bQ/S)_{\bar r}$ is represented by the $1$-form \eqref{eq:forms} since  \eqref{eq:isooverintegers} extends  \eqref{eq:isooverrationals}. The construction of 
 $[\omega^{(-r)}_i]$ is similar.
 Finally, the explicit formula for the Gauss-Manin connection follows from 
Proposition \ref{pr:expliciteformulafortheKZconn}. 
\end{proof}

\subsection{Regular forms}

\begin{lem}
\label{lem hol}
Let $1\leq r<q$ be an integer congruent to $\bar r$ modulo $q$.  Then the differential forms
\bean
\label{c3}
\mu^{(r)}_i = \frac{x^{i-1}dx}{y^r}\,,\qquad  i=1,\dots, \Big[\frac {nr}q\Big],
\eean
form a basis for the $\cO_S$-module  $F^1_{\bar r}$.
Hence $F^1_{\bar r}$ is a free module of rank $\big[\frac {nr}q\big]$.

\end{lem}

\begin{proof} 

Easy calculations in local coordinates show that every such 
$\mu^{(r)}_i $ is regular. It suffices to check that they form a basis for the fiber of  $F^1_{\bar r}$ over each closed point of $S$.
Clearly the forms
$\mu^{(r)}_i $ generate a subspace of the fiber  of dimension  $\big[\frac {nr}q\big]$.  The total number of them equals
\bea
\sum_{r=1}^{q-1}\Big[\frac {nr}q\Big] = \frac{(n-1)(q-1)}2 = g.
\eea
This proves the lemma.
\end{proof}

\begin{lem}
\label{lem 5.2} Let $1\leq r<q$ be an integer congruent to $\bar r$ modulo $q$. Then,
for $i=1,\dots,n$  we have the following identity in $ H^1_{dR}(X/S)_{\bar r}\otimes \bZ[N_{r,n,q}^{-1}]$
\bean
\label{dmg}
\nabla_{\!\!\frac{\partial}{\partial z_i}}
[\mu_k^{(r)}]= \frac rq\big([\mu_{k-1}^{(r)}] + z_i[\mu_{k-2}^{(r)}]+\dots + z_i^{k-2}[\mu_1^{(r)}] + z_i^{k-1}[\om_i ^{(r)}]\big).
\eean

\end{lem}

\begin{proof}
Since the restriction map $H^1_{dR}(X/S) \to H^1_{dR}(X^\circ/S)$ is injective it suffices to prove \eqref{dmg} in $H^1_{dR}(X^\circ/S)$. To this end, denote by $\theta_i$
the vector field on $X^\circ$ defined by $\Lie_{\theta_i} (z_j) =\delta_{ij}$, $1\leq j  \leq n$, $\Lie_{\theta_i} (x)= 0$, and $\Lie_{\theta_i}(y)= \frac{-1}{q}\frac{y}{x-z_i}$. 
Then, we have
\bea
\nabla_{\!\!\frac{\partial}{\partial z_i}}[\mu_k ^{(r)}]&=& 
\left[\Lie_{\theta_i} \left(\frac{x^{k-1}dx}{y^r}\right)\right] 
= \left[\frac rq\left(\frac{x^{k-1}dx}{(x-z_i)y^r}\right)\right]
\\
&=&
 \left[\frac rq\left(\frac{(x^{k-1} -z_ix^{k-2} + z_ix^{k-2} -z_i^2x^{k-3} +\dots + z_i^{k-1})
dx}{(x-z_i)y^r}\right)\right]
\\
&=& 
\left[\frac rq\big(\mu_{k-1}^{(r)} + z_i\mu_{k-2}^{(r)}+\dots 
+ z_i^{k-2}\mu_1^{(r)} + z_i^{k-1}\om_i^{(r)} \big)\right].
\eea
\end{proof}

\begin{cor}
\label{cor 5.3}
We have
\bean
\label{dmg1}
\nabla_{\!\!\frac{\partial}{\partial z_i}}
[\mu_1^{(r)}]\ =\ \frac rq\,[\om_i ^{(r)}].
\eean

\end{cor}

\subsection{Poincar\'e pairing}

\begin{thm}
\label{thm Poi}
Let $r=1,\dots,q-1$,  $k=1,\dots, n$,
$j=1,\dots, \big[\frac{n(q-r)}q\big]$.
Then the Poincar\'e pairing $([\om_k^{(r)}], [\mu_j^{(q-r)}])$
 is given by the formula
\bean
\label{ommu}
([\om_k^{(r)}], [\mu_j^{(q-r)}]) \,=\, 
-\,\frac{q}{rC_k(z)}\,z_k^{j-1}\,,
\eean
where $C_k(z) = (z_k-z_1)\dots (z_k-z_{k-1})(z_k-z_{k+1}) \dots (z_k-z_{n})$.

\end{thm}

\begin{proof}

Recall that $X$ is a compactification  of the affine curve  defined by the equation $y^q =\prod_{a=1}^n(x-z_a)$.
Define an open cover of $X$,
\bea
U_1 = X- D_k,\qquad 
U_2 = X - {\bigcup}_{{i=1
\atop i\ne k}}^{n} D_i.
\eea
Elements of $ H^1_{dR}(X)$ can be represented by cochains
 $(\al,\beta,\ga)$ consisting of a regular 1-form $\al$ on $U_1$,
 a regular 1-form $\beta$ on $U_2$, and a regular function $\ga$ on $U_1\cap U_2$
 such that $d \ga= \al-\beta$. The Poincare pairing of two cochains
 $(\al_1,\beta_1,\ga_1)$ and  $(\al_2,\beta_2,\ga_2)$ can be computed  as
 \bea
  \Res_{D_k} \big(\ga_1\al_2 - \ga_2\al_1\big).
\eea
The cohomology classes  $\om_1^{(r)}$ and $\mu_k^{(q-r)}$ are represented by the  cochains:
\bea
\Big(\om_k^{(r)}, -\sum_{{i=1\atop i\ne k} }^{n} \om_i^{(r)}, -\frac q{ry^r}\Big), \qquad
(\mu_j^{(q-r)}, \mu_j^{(q-r)}, 0).
\eea
Hence
\bea
([\om_k^{(r)}], [\mu_j^{(q-r)}]) = \Res_{D_k} \Big(\frac{-q}{ry^r} \cdot \frac{x^{j-1}dx}{y^{q-r}}\Big)
=\frac{-q}{rC_k(z)}\,z_k^{j-1}\,.
\eea
\end{proof}

\subsection{Fundamental exact sequence} Let $p$ be a prime number that does not divide $nq$. 
For the rest of this paper we let $S_{\F_p}=S\times \Spec \bF_p$ (resp. $X_{\F_p}=X\times \Spec \bF_p$).

The  relative de Rham cohomology
$H^1_{dR}(X_{\F_p}/S_{\F_p})$ 
is equipped with the Gauss-Manin  connection $\nabla$,  Hodge and conjugate filtrations (\cite[Section 2.3]{katz}).
 For dimension reason, the conjugate spectral sequence degenerates at the $E_1$-page and 
yields a short exact sequence of vector bundles:
\begin{equation}\label{eq:fes_preliminary}
  0\to F^*_{S_{\F_p}} H^1(S_{\F_p}, \cO_{X_{\F_p}}) \rar{F^*_{X_{\F_p}/S_{\F_p}}} 
 H^1_{dR}(X_{\F_p}/S_{\F_p}) \rar{C}  F^*_{S_{\F_p}} H^0(S_{\F_p}, \Omega^1_{X_{\F_p}/S_{\F_p}}) \to 0. 
\end{equation}
 Recall the construction of the maps in \eqref{eq:fes_preliminary}. For every function $f\in \cO_{X_{\F_p}}$, we have $d(f^p)=0$. Thus the $p$-th power map $F^*_{abs}$
induces a morphism of complexes $\cO_{X_{\F_p}} \rar{F^*_{abs}} (\cO_{X_{\F_p}}\rar{d} \Omega^1_{X_{\F_p}/S_{\F_p}})$. Passing to cohomology we get an $\cO_{S_{\bF_p}}$-linear 
map $H^1(S_{\F_p}, \cO_{X_{\F_p}}) \to F_{*, S_{\F_p}} H^1_{dR}(X_{\F_p}/S_{\F_p})$ which by adjunction gives $F^*_{X_{\F_p}/S_{\F_p}}$. Note that since the composite
$H^1(S_{\F_p}, \cO_{X_{\F_p}}) \mono F^*_{S_{\F_p}} H^1(S_{\F_p}, \cO_{X_{\F_p}}) \to  H^1_{dR}(X_{\F_p}/S_{\F_p})$ factors through the absolute de Rham cohomology $H^1_{dR}(X_{\F_p})$ its image consists of cohomology classes which are flat with respect to $\nabla$.
The second morphism  in \eqref{eq:fes_preliminary} arises from the Cartier isomorphism 
$$
C: \cH^1(\Omega_{X_{\F_p}/S_{\F_p}}^\bullet, d) \iso \Omega^1_{X_{\F_p}^{\prime}/S_{\F_p}}\, ,
$$
where $X_{\F_p}^{ \prime}: = X_{\F_p} \times_{S_{\F_p}, F_{S_{\F_p}}} S_{\F_p}$ denotes the Frobenius twist $X_{\F_p}$.

Since the Hodge spectral sequence also degenerates at the $E_1$-page, we have 
$$
H^1(S_{\F_p}, \cO_{X_{\F_p}}) \iso H^1_{dR}(X_{\F_p}/S_{\F_p})/F^1, \qquad  H^0(S_{\F_p},  \Omega^1_{X_{\F_p}/S_{\F_p}}) \iso F^1,
$$
where $F^1 \subset  H^1_{dR}(X_{\F_p}/S_{\F_p}) $ is the Hodge filtration. 
Thus, the above short exact sequence takes the form:
 \begin{equation}
 \label{eq:fundamental_exact_sequence}
 0\to F^* _{S_{\F_p}}
 (H^1_{dR}(X_{\F_p}/S_{\F_p})/F^1) \xrightarrow{F_{X_{\F_p}/S_{\F_p}}^*}
  H^1_{dR}(X_{\F_p}/S_{\F_p}) \rar{C} F^* _{S_{\F_p}}F^1 \to 0.
 \end{equation}
 
We remark that morphisms $C$ and $F^*_{X_{\F_p}/S_{\F_p}}$ are adjoint with respect to the Poincar\'e pairing:  
\begin{equation}\label{eq:adjunction_of_F_and_C}
    (F^*_{X_{\F_p}/S_{\F_p}}(\alpha), \beta)=  (\alpha, C(\beta)), \quad \alpha\in   F^* _{S_{\F_p}}
 (H^1_{dR}(X_{\F_p}/S_{\F_p})/F^1), \; \; 
    \beta \in H^1_{dR}(X_{\F_p}/S_{\F_p}).
\end{equation}
 
 \smallskip

Recall that, for every vector bundle  $V$ over $S_{\F_p}$ its Frobenius pullback $F^* _{S_{\F_p}}  V$
 is equipped with the Frobenius descent connection characterized by the property that the pullback of any local section of $V$ is flat.
If we equip the boundary terms of sequence (\ref{eq:fundamental_exact_sequence}) with the Frobenius descent connection and the middle term with the Gauss-Manin  connection, then all the differentials in  (\ref{eq:fundamental_exact_sequence}) are flat.

\smallskip
The exact sequence \eqref{eq:fundamental_exact_sequence} splits into isotypic components
\bean
\label{eq:fundamental_exact_sequence-isotypic} \qquad
 0\to (F^* _{S_{\F_p}}(H^1_{dR}(X_{\F_p}/S_{\F_p})/F^1))_{\bar r} \rar{F_{X_{\F_p}/S_{\F_p}}^*} 
 H^1_{dR}(X_{\F_p}/S_{\F_p})_{\bar r} \rar{C} (F^* _{S_{\F_p}}F^1)_{\bar r} \to 0.
\eean
Note that 
\bean
\label{r/p}
 (F^* _{S_{\F_p}}F^1)_{\bar r} =  F^* _{S_{\F_p}} (F^1_{{\bar r/p}}) \,.
 \eean

\subsection{Kodaira-Spencer map and $p$-curvature.}

Recall the Kodaira-Specer map
\bea
 \text{KS}_\nabla: F^1 \to H^1_{dR}(X_{\F_p}/S_{\F_p})/F^1  \otimes \Omega^1_{S_{\F_p}}\,.
 \eea
 By definition it sends the class $\nu$ of a regular one-form to
 $\sum_{i=1}^n \nu_i \ox dz_i$\,,\,
where $\nu_i$ is the projection of  $\nabla_{\frac{\der}{\der z_i}} \nu$ to 
$H^1_{dR}(X_{\F_p}/S_{\F_p})/F^1 $. Assume that $p$ does not divide $qnN$, where $N$ is defined by formula \eqref{eq:N}.  Then using Lemma \ref{lem 5.2} we obtain a formula
for the Kodaira-Specer map:
\bean
\label{KS}
\text{KS}_\nabla\,:\, [\mu^{(r)}_k] \ \mapsto \  \frac{r}{q}  \sum_{i=1}^n  z_i^{k-1}\overline{\om^{(r)}_i }\ox dz_i\,,
\qquad  k=1,\dots, \Big[\frac {nr}q\Big],
\eean
where $\overline{\om^{(r)}_i}$ is the projection of $[\om^{(r)}_i]$ to 
$H^1_{dR}(X_{\F_p}/S_{\F_p})/F^1$\,.

Let 
\bea
\Psi: H^1_{dR}(X_{\F_p}/S_{\F_p}) \to H^1_{dR}(X_{\F_p}/S_{\F_p}) \otimes F^*_{S_{\F_p}} \Omega^1_{S_{\F_p}}
\eea
 be the $p$-curvature of the Gauss-Manin  connection. 
 The $p$-curvature operator $\Psi$ is a product of three morphisms
\bean
\label{facto2}
\Psi = F_{X_{\F_p}/S_{\F_p}}^* \circ \bar \Psi \circ C,
\eean
where
 \bean
\label{facto}
\bar \Psi:  F^* _{S_{\F_p}}F^1 
 \to F^* _{S_{\F_p}}((H^1_{dR}(X_{\F_p}/S_{\F_p})/F^1 ) 
\otimes \Omega^1_{S_{\F_p}}).
\eean 
The factorization holds 
since the  $p$-curvature  of  the Frobenius descent connection equals $0$.
 
\smallskip

The Katz $p$-curvature formula (\cite[Theorem 3.2]{katz}) says that 
\begin{equation}
\label{eq:Katz-p}
 \bar \Psi \, =\, -\, F^*_{S_{\F_p}}(\text{KS}_\nabla),  
\end{equation}
where  $\on{KS}_\nabla: F^1 \to H^1_{dR}(X_{\F_p}/S_{\F_p})/F^1  \otimes \Omega^1_{S_{\F_p}}$ 
is the Kodaira-Spencer map. 

\medskip

For $k=1,\dots,n$,  denote by
\bean
\label{Psi k}
\bar \Psi_k\,:\,  F^* _{S_{\F_p}}F^1  \to F^* _{S_{\F_p}}((H^1_{dR}(X_{\F_p}/S_{\F_p})/F^1 ) 
\eean
the contraction of $\bar \Psi$ and the vector field $\frac\der{\der z_k}$.
Define the $p$-curvature operators
\bean
\label{Psi k1}
\Psi_k = F_{X_{\F_p}/S_{\F_p}}^* \circ \bar \Psi_k\circ C \,:\, H^1_{dR}(X_{\F_p}/S_{\F_p}) \to  
H^1_{dR}(X_{\F_p}/S_{\F_p}).
\eean

\begin{lem}
\label{lem CC}

For $k,\ell\in\{1,\dots,n\}$, the $p$-curvature operators have the property
\bean
\label{Cc}
\Psi_k \Psi_\ell = 0.
\eean

\end{lem}

\begin{proof} We have 
$\Psi_k \Psi_\ell =  F_{X_{\F_p}/S_{\F_p}}^* \circ \bar \Psi_k \circ C \circ  F_{X_{\F_p}/S_{\F_p}}^*
 \circ \bar \Psi_\ell \circ C$
and $ C\circ  F_{X_{\F_p}/S_{\F_p}}^*=0$ since sequence  \eqref{eq:fundamental_exact_sequence-isotypic} 
is exact.
\end{proof}

\begin{lem}\label{lm:image_of_p_curvature_is_big} Let $1\leq r<q$ be an integer congruent to $\bar r$ modulo $q$.
Assume that $p$ does not divide $qnN$, where $N$ is defined by formula \eqref{eq:N}.
Then  the sections $\bar{\Psi}_k ([\mu_1^{(r)}]),$
$k=1,\dots,n$, span $(F^* _{S_{\F_p}}(H^1_{dR}(X_{\F_p}/S_{\F_p})/F^1))_{\bar r}$\,.
 \end{lem}
\begin{proof}
Using  the Katz $p$-curvature formula \eqref{eq:Katz-p} and Colloary \ref{cor 5.3}, we have
$$  \bar{\Psi}_k ([\mu_1^{(r)}]) =  - \ \frac rq\,[\om_k ^{(r)}].   $$
\end{proof}
\subsection{Flat sections of $H^1_{dR}(X_{\F_p}/S_{\F_p})_{\bar r}$.}
We are ready to prove the main result of this section.
\begin{thm}\label{th:flat_sections}
    Let $1\leq r<q$ be an integer congruent to $\bar r$ modulo $q$.
    Assume that $p$ does not divide $N_{r,q,n}$ and $(F^* _{S_{\F_p}}F^1)_{\bar r}\ne 0$. Then the morphism $C$ in \eqref{eq:fundamental_exact_sequence-isotypic} 
induces an isomorphism on flat functionals
\begin{equation}\label{eq:mth_on_flat_funct}
 \Hom _{\cO_{S_{\bF_p}}} ((F^* _{S_{\F_p}}F^1)_{\bar r}, \cO_{S_{\bF_p}})^{\nabla=0} \iso  \Hom_{\cO_{S_{\bF_p}}}(H^1_{dR}(X_{\F_p}/S_{\F_p})_{\bar r}, \cO_{S_{\F_p}})^{\nabla=0}.
\end{equation}
Also the right-hand side of \eqref{eq:mth_on_flat_funct} is isomorphic to  $\Hom_{\cO_{S_{\bF_p}}}(F^1_{\bar r/p}, \cO_{S_{\bF_p}}) $.

\end{thm}
\begin{proof}
    Note that, for every flat map
    $f\colon H^1_{dR}(X_{\F_p}/S_{\F_p})_{\bar r}\to \cO_{S_{\bF_p}}$, every $s\in H^1_{dR}(X_{\F_p}/S_{\F_p})_{\bar r}$, 
    and every $1\leq k\leq n$, we have
$$
f\left(\Psi_{k} (s)\right)= \frac{\partial ^p f(s)}{\partial z_i^p }=0.
$$
Thus $f$ vanishes on the image of $\Psi_{k}$. By Lemma \ref{lm:image_of_p_curvature_is_big} classes $\Psi_k ([\mu_1^{(r)}]),$
$k=1,\dots,n$, span the image of $\big(F^* _{S_{\F_p}}(H^1_{dR}(X_{\F_p}/S_{\F_p})/F^1)\big)_{\bar r}$.
Hence $f$ vanishes on this space. Consequently, the homomorphism $C$ identifies the space of flat maps 
$H^1_{dR}(X_{\F_p}/S_{\F_p})_{\bar r}\to \cO_{S_{\bF_p}}$
with space of flat maps $F^* _{S_{\F_p}}(F^1_{\bar r/p}) \to  \cO_{S_{\bF_p}}$.
The second assertion of part (2) is a consequence of the Cartier descent (\cite{katz70}).   
\end{proof}
\subsection{A formula for the Cartier operator.} We shall explain a formula for the map $C$ in \eqref{eq:fundamental_exact_sequence-isotypic} that will be used in Section 
\ref{sec 5}. Observe that $C$ commutes with restriction to any open subset of $X$ invariant under $\mu_q$. In particular, we have a commutative diagram
\begin{equation}\label{dia:restriction}
\begin{tikzcd}
 H^1_{dR}(X_{\F_p}/S_{\F_p})_{\bar r}   \arrow[r, "C"] \arrow[d]& \big(F^* _{S_{\F_p}} H^0(X, \Omega^1_{X_{\F_p}/S_{\F_p}})\big)_{\bar r}\arrow[d, hook] \\
    H^1_{dR}(X_{\F_p}^\circ/S_{\F_p})_{\bar r}  \arrow[r, "C" , "\cong"']& \big(F^* _{S_{\F_p}} H^0(X^\circ, \Omega^1_{X_{\F_p}^\circ/S_{\F_p}})\big)_{\bar r}
\end{tikzcd}
\end{equation}
 where the bottom horizontal arrow is an isomorphism and the right downward arrow is injective. We shall compute $C([\omega^{(r)}_i])\in  \big(F^* _{S_{\F_p}} H^0(X^\circ, \Omega^1_{X_{\F_p}^\circ/S_{\F_p}})\big)_{\bar r}$, where $\omega^{(r)}_i$ are
 $1$-forms on $X^\circ$ defined in \eqref{eq:forms}. 
 \begin{lem}\label{lem:computing_C}
Write $r=pa - q \tilde h$, where $a$ and $\tilde h$ are integers with $0< \tilde h \leq p$. For $1\leq i \leq n$, define  polynomials $Q^{(k)}_i(z,\tilde h)$ in $z_1, \dots, z_n$
by the formula
$$\frac{P(x, z)^{\tilde h}}{x-z_i}=\sum_{0\leq k< n\tilde h} Q^{(k)}_i(z,\tilde h) x^k,$$
where $P(x,z)=\prod_{i=1}^n(x-z_i)$.  
Then, we have
$$C([\omega^{(r)}_i]) = \sum_{1\leq l \leq \big[\frac{n\tilde h}{p}\big]} \mu^{(a)}_{l}  \otimes Q^{(pl-1)}_i(z,\tilde h) \in H^0(X^\circ, \Omega^1_{X_{\F_p}^ \circ/S_{\F_p}})_{\bar a}  \otimes _{\cO_{S_{\F_p}}, F_{S_{\F_p}}}  \cO_{S_{\F_p}}.$$
 \end{lem}
\begin{proof}
We have
$$C\Big(\frac{dx}{y^r(x-z_i)}\Big)=  C\Big(\frac{P(x, z)^{\tilde h} dx}{y^{pa}(x-z_i)}\Big) = \sum_{0\leq k< n\tilde h} C(y^{-ap} Q^{(k)}_i(z,\tilde h) x^k dx).$$
It remains to observe that $C( y^{-ap} Q^{(k)}_i(z,\tilde h) x^k dx)=  y^{-a}x^{l-1}dx \otimes Q^{(pl-1)}_i(z,\tilde h)$, for $k$ of the form $pl-1$, and $0$ otherwise.
\end{proof} 
If $p$ does not divide the integer $N_{r,n,q}$ defined by\eqref{eq:N} then Theorem \ref{th:basis} (iii) together with Lemma \ref{lem:computing_C} yield a concrete description 
of the upper horizontal map in \eqref{dia:restriction}.
\section{Knizhnik-Zamolodchikov equations} 
\label{sec 5}
In this section we use the curve introduced in Section \ref{sec 4} together with the Katz formula for the p-curvature of the Gauss-Manin connection 
to give a (second) proof of Theorem \ref{thm:main3_introduction}. In particular, we describe all solutions to the KZ equations in characteristic $p$, for $\bar{h}\in \F_p$.
We also compute explicitly the p-curvature of the KZ connection.   

\subsection{All solutions of KZ  equations are $p$-hypergeometric}
\label{sec 5.1}

Fix an integer $n>1$.
Consider   the trivial vector bundle $\cO_{S_{\bF_p}}^{\oplus n}=\bigoplus_{i=1}^{i=n} \cO_{S_{\bF_p}}
 \cdot e_i$ over $S_{\bF_p}$
 with a connection defined by  the formula 
$$
\nabla= d + h \sum_{1\leq i\ne j \leq n} \frac{\Omega_{ij}}{z_i - z_j}dz_i\,.
$$
Here $h \in \F_p$ and  $\Omega_{ij}$ are the $n\times n$ matrices defined in \eqref{Omij}.
The connection  is integrable: $\nabla^2=0$.

\smallskip

Let $\mc V$ be the kernel of the addition map $\Sigma: \cO_{S_{\bF_p}}^{\oplus n}\to \cO_{S_{\bF_p}}$.
Thus  $\mc V$ is a (trivial) rank $n-1$ subbundle of $\cO_{S_{\bF_p}}^{\oplus n}$.

\smallskip

The connection $\nabla$ preserves $\mc V$  and the restriction of 
$\nabla$ to $\mc V$ is the KZ connection $\nabla^{\on{KZ},h}$ in Section
 \ref{sec 1.1}. Set $\mc V_{h}:= (\mc V, \nabla^{\on{KZ},{h}})$.

\smallskip

The next theorem is a reformulation of Theorem \ref{thm:main3_introduction} from the Introduction. In particular, 
it describes flat sections of $\mc V_{{h}}$ in characteristic $p>0$.

\begin{thm}
\label{thm 7}

 Let $p$ be a prime number that does not divide $n$, $0\ne \bar h\in \bF_p$.
 \begin{enumerate}
     \item Let $0\leq \tilde h < p-1$ be the integer congruent to $\bar h$ modulo $p$.
Set $d=\big[\frac{n \tilde h}{p}\big]$,  if $\big[\frac{n \tilde h}{p}\big]>0$,
 and $d=n-1$, if $\big[\frac{n \tilde h}{p}\big]=0$.
Then the space $\Gamma(S_{\bF_p}, \mc V_{h})^{\nabla=0}$ of global flat sections of the bundle $\mc V$ is a free module over the algebra 
$\bF_p [z_i^p, (z_i^p-z_j^p)^{-1}, 1\leq i< j \leq n]$ of rank $d$. Moreover, if $ \big[\frac{n\tilde h}{p} \big]>0$,
 then    $\Gamma(S_{\bF_p}, \mc V_{h})^{\nabla=0}$
is freely generated over $\bF_p [z_i^p, (z_i^p-z_j^p)^{-1}, 1\leq i< j \leq n]$ 
 by $p$-hypergeometric sections from \eqref{p hyp}.
\item Let $\cU_{h}\subset \mc V_{h}$ be the $\cO_{S_{\bF_p}}$-submodule spanned by  the $p$-hypergeometric sections. Similarly, we define $\cU_{-h}\subset \mc V_{-h}$.
Then $\cU_{h}$ and $\cU_{-h}$ are subbundles satisfying the following Lagrangian property:
the Shapovalov form $ \mc V_{h}\otimes  \mc V_{-h}\to 
\cO_{S_{\bF_p}}$  (see  \eqref{eq:dualitySha}) restricts to the zero map on $ \mc U_{h}\otimes  \mc U_{-h}$ and induces an isomorphism
\begin{equation}
     \big(\cV_{h}/\mc U_{h}\big)^*\iso \mc U_{-h}.
\end{equation}

 \end{enumerate}
 \end{thm}

\begin{proof}  

Recall that $\mc V_{ h}$ is isomorphic to the local system dual to $E_{ h}$, see \eqref{eq:dualitySha} and \eqref{eq:explicit_formula_for_connection}. We shall use Theorem \ref{th:flat_sections} to analyze flat sections of  $E_{ h}^*$.

Pick positive integers $q$, $a$  such that 
 $q>n$ and $q\tilde h +1 = a p$. 
  Observe that 
 \begin{equation}
     \left[\frac{n\tilde h}{p} \right]=\left[\frac{n a}{q} \right].
\end{equation}
Indeed, we have
\bea
\frac{n\tilde h}p + \frac{n}{pq} = \frac{na}q
\quad \text{and} \quad 
1+\left[\frac{n\tilde h}p\right]
-\frac{n\tilde h}p\geq \frac1p\,,
\quad
\text{while}\quad
\frac{n}{pq}< \frac 1p\,.
\eea

Consider the curve $X$ over $S$ defined in Section \ref{ss:curve}.  For this $q$ and $r=1$, the number $N$, introduced in \eqref{eq:N}, equals $-n$. Thus, by Theorem \ref{th:basis}, the local system $E_{ h}$ is isomorphic to $H^1_{dR}(X_{\bF_p}/S_{\bF_p})_{\bar 1}$.  We have  the  following exact sequence of flat bundles:
 \bean
\label{ies}
 0\to (F^* _{S_{\bF_p}}(H^1_{dR}(X_{\bF_p}/S_{\bF_p})/F^1))_{\bar 1} \rar{F_{X_{\bF_p}/S_{\bF_p}}^*} 
 H^1_{dR}(X_{\bF_p}/S_{\bF_p})_{\bar 1} \rar{C} (F^* _{S_{\bF_p}}F^1)_{\bar 1} \to 0,
\eean
see \eqref{eq:fundamental_exact_sequence-isotypic}, and an isomorphism  $(F^* _{S_{\bF_p}}F^1)_{\bar 1} =  F^* _{S_{\bF_p}} (F^1_{\bar a})$.
By Lemma \ref{lem hol} the rank of  $F^1_{\bar a}$ equals $ \big[\frac{n a}{q} \big]$. 

If $\big[\frac{n \tilde h}{p}\big] = \big[\frac{n a}{q} \big]= 0$, we have an isomorphism $$E_{ h} \iso (F^* _{S_{\bF_p}}(H^1_{dR}(X_{\bF_p}/S_{\bF_p})/F^1))_{\bar 1}.$$ In particular, using the Cartier descent, the space of flat sections
of $\mc V$ is identified with  $$\Hom_{S_{\bF_p}} ( (H^1_{dR}(X_{\bF_p}/S_{\bF_p})/F^1)_{\bar a}, \cO_{S_{\bF_p}})\iso  F^1_{- \bar a}$$ which is a free $\cO_{S_{\bF_p}}$-module (of rank $n-1$) as desired. 

Assume that  $\big[\frac{n \tilde h}{p}\big] >0$. Then, by Theorem \ref{th:flat_sections},  the space  $\mc V^{\nabla=0}$ is identified with 
$\Hom_{\cO_{S_{\bF_p}}}(F^1_{\bar a}, \cO_{S_{\bF_p}})$. Consider the basis for the latter which is dual to the basis 
$\mu^{(a)}_i$, $i=1,\dots, \Big[\frac {na}q\Big]$. Then by Lemma \ref{lem:computing_C} (see also \cite[Theorem 6.2]{SlV}) the image in $\mc V^{\nabla=0}$ of 
the elements of that dual basis are precisely the
$p$-hypergeometric sections. This proves the first part of the theorem.

To prove part (2) recall from Theorem \ref{th:basis} the isomorphism  $E_{- h}\iso H^1_{dR}(X_{\bF_p}/S_{\bF_p})_{-\bar 1}$. Consider 
the  exact sequence of flat bundles:
 \begin{equation}
\label{ies_dual}
 0\to (F^* _{S_{\bF_p}}(H^1_{dR}(X_{\bF_p}/S_{\bF_p})/F^1))_{- \bar 1} \rar{F_{X_{\bF_p}/S_{\bF_p}}^*} 
 H^1_{dR}(X_{\bF_p}/S_{\bF_p})_{- \bar 1} \rar{C} (F^* _{S_{\bF_p}}F^1)_{- \bar 1} \to 0.
 \end{equation}
Using Lemma \ref{lem:computing_C},  the isomorphism  $\cV_{h}\iso E_{ h}^*$ identifies  $\cU_{h}\subset \cV_ h$ 
with $\Big((F^* _{S_{\bF_p}}F^1)_{\bar 1}\Big)^* \subset E_{ h}^*$ and,
similarly, $\cU_{-h} \iso  \Big((F^* _{S_{\bF_p}}F^1)_{-\bar 1}\Big)^*$. By Theorem \ref{th:basis} the Shapovalov form $ \mc V_{h}^*\otimes  \mc V_{-h}^*\to 
\cO_{S_{\bF_p}}$ equals the Poincar\'e form $ H^1_{dR}(X_{\bF_p}/S_{\bF_p})_{\bar 1}\otimes  H^1_{dR}(X_{\bF_p}/S_{\bF_p})_{- \bar 1}\to \cO_{S_{\bF_p}}$ up to a non-zero constant factor. Consequently, the Lagrangian property 
from part (2) is equivalent to the following: the  Poincar\'e form restricts to zero on  $(F^* _{S_{\bF_p}}(H^1_{dR}(X_{\bF_p}/S_{\bF_p})/F^1))_{\bar 1}\otimes (F^* _{S_{\bF_p}}(H^1_{dR}(X_{\bF_p}/S_{\bF_p})/F^1))_{- \bar 1}$ and induces an isomorphism 
$$  (F^* _{S_{\bF_p}}(H^1_{dR}(X_{\bF_p}/S_{\bF_p})/F^1))_{\bar 1}\iso \Big((F^* _{S_{\bF_p}}F^1)_{- \bar 1}\Big)^*.$$
But this follows from formula \eqref{eq:adjunction_of_F_and_C} and the Serre duality.
\end{proof}

Fix a point  $a=(a_1,\dots, a_n) \in S(\bar{\bF}_p)$. We refer to $\cO_{S_{\bF_p}}$-linear maps $\mc V^*\to \bar{\bF}_p[[z_1-a_1, \dots ,z_n-a_n]]$
as formal sections of $\cV$ at $a$.
\begin{cor}
\label{cor form}
Let $p$ be a prime number that does not divide $n$, $0\ne  h\in \bF_p$,  and  let $0\leq \tilde h < p-1$ be the integer congruent to 
$ h$ modulo $p$.
Assume $\big[\frac{n \tilde h}{p}\big]>0$. 
Then the space of formal flat sections of $\mc V$ is a free module of rank $\big[\frac{n \tilde h}{p}\big]$ over the algebra 
$\bar{\bF}_p[[(z_1-a_1)^p, \dots, (z_n-a_n)^p]]$.
The $p$-hypergeometric sections 
form a basis of this module.
\end{cor}
\begin{proof} This follows from the theorem above combined with Lemma \ref{lm:base_change_to_disk}. 
\end{proof}

\subsection{$p$-curvature operators $\Psi_k$}
\label{sec p-cu}

Use the notations of Theorem \ref{thm 7} and its proof. 
By formula \eqref{Psi k1}, 
\bean
\label{tr}
\Psi_k = F_{X_{\bF_p}/S_{\bF_p}}^* \circ \bar \Psi_k \circ C\,:\,H^1_{dR}(X_{\bF_p}/S_{\bF_p})_{\bar 1}
\to H^1_{dR}(X_{\bF_p}/S_{\bF_p})_{\bar 1}\,.
\eean
Apply this composition  to  elements $\om_i^{(1)}$ generating $H^1_{dR}(X_{\bF_p}/S_{\bF_p})_{\bar 1}$.
Using Lemma \ref{lem:computing_C}, we have
\bean
\label{re}
C\big(\om_i^{(1)}) = \sum_{\,\ell=1}^{\big[\frac{n\tilde h}p\big] }  \mu_\ell^{(a)} \otimes  Q^{(\ell p-1)}_i(z, \tilde h),
\eean
where $Q^{(\ell p-1)}_i(z,\tilde h)$ are coordinates of the $p$-hypergeometric solutions in Section \ref{sec 1}.
\smallskip

Note that if $\big[\frac{n\tilde h}p\big]=0$, then all $p$-curvature operators $\Psi_k$ are equal to zero.
\smallskip

We apply the operator $\bar \Psi_k$ to the expression in  \eqref{re} and obtain
\bean
\label{PsC}
\bar\Psi_k(C\big(\om_i^{(1)})) = - \frac aq \; \overline{\om_k^{(a)}} \otimes \left(\sum_{\,\,\ell=1}^{\big[\frac{n\tilde h}p\big] } Q^{(\ell p-1)}_i(z,  h)  z_k^{{p}(\ell-1)} \right)\,,
\eean
see formula \eqref{KS}. 
\smallskip

Formula \eqref{PsC} shows that the rank of $\Psi_k$ is not greater than 1.

\smallskip

We compute the last operator $ F_{X/S}^* $ in formula \eqref{tr} 
using the adjunction property \eqref{eq:adjunction_of_F_and_C},
\bean
\label{lf}
& &  \left(F_{X/S}^*\Big(\,\overline{\om_k^{(a)}}\,\Big), \om _j^{(-1)}\right) = \left( \om_k^{(a)}, C(\om _j^{(-1)})\right)=
\\
\notag
& &
\sum_{\,\,\,m = 1}^{\big[\frac{n(p-\tilde h)}p\big]}
Q_j^{(mp-1)}(z, -h)  (\om_k^{(a)},\mu_m^{(q-a)})^p = \frac{-q}{aC_k(z)^p} \sum_{\,\,\,m = 1}^{\big[\frac{n(p-\tilde h)}p\big]}
Q_j^{(mp-1)}(z, - h)  z_k^{p(m-1)}\,,
\eean
see formula \eqref{ommu}.

\smallskip

Notice that if $\big[\frac{n\tilde h}p\big]=n-1$, then $\big[\frac{n(p-\tilde h)}p\big]=0$
and hence all the $p$-curvature operators $\Psi_k$ are equal to zero.
\smallskip

This calculation gives the following description of the $p$-curvature operators.

\begin{thm}
\label{thm pPk}
Let $0<\big[\frac{n\tilde h}p\big]<n-1.$
Then for $k=1,\dots,n$,  the operator $\Psi_k$ is of rank 1 and
\bean
\label{cU}
{}
\\
\notag
\left( \Psi_k(\om_i^{(1)}), \om _j^{(-1)}\right)
=
\frac{1}{C_k(z)^p}
 \left(\sum_{\,\,\ell=1}^{\big[\frac{n\tilde h}p\big] } Q^{(\ell p-1)}_i(z,h)  z_k^{{p}(\ell-1)} \right)
\left(\sum_{\,\,\,m = 1}^{\big[\frac{n(p-\tilde h)}p\big]}
Q_j^{(mp-1)}(z, -h) \,z_k^{p(m-1)}\right)
\eean
for all $i,j\in\{1,\dots,n\}$.
 \qed
\end{thm}


\begin{cor}
\label{cor 55}
An element $v= \sum_{n=1}^n v_i \om^{(1)}_i$ lies in the kernel of $\Psi_k$ if and only if
\bean
\label{ker k}
\sum_{i=1}^n v_i\left(\sum_{\,\,\ell=1}^{\big[\frac{n\tilde h}p\big] } Q^{(\ell p-1)}_i(z, h)  z_k^{{p}(\ell-1)} \right) = 0.
\eean
\qed

\end{cor}

\begin{cor}
\label{cor 5.4}

We have
\bean
\label{cU 5.4}
{}
\\
\notag
\Psi_k(\om_i^{(1)})
=
\frac{-h}{C_k(z)^p}
 \left(\sum_{\,\,\ell=1}^{\big[\frac{n\tilde h}p\big] } Q^{(\ell p-1)}_i(z,h)  z_k^{{p}(\ell-1)} \right)
\sum_{\,\,\,m = 1}^{\big[\frac{n(p-\tilde h)}p\big]}\sum_{a=1}^n 
Q_a^{(mp-1)}(z, -h) \,z_k^{p(m-1)} \om^{(1)}_a 
\eean
for all $i\in\{1,\dots,n\}$.

\end{cor}

Corollary \ref{cor 5.4} follows from
the formula $(\om_a^{(1)},\om_b^{(-1)}) = -\frac 1h\left(\delta_{ab}-\frac 1n\right)$, see \eqref{eq:poincare_duality_expl}.



\subsection{Identification with $p$-curvature operators on $\mc V_{-h}$ and $\mc V_h$}

Denote by $\Psi^{\on{KZ}, g}_k(z)$ the $k$-th $p$-curvature operator of  $\mc V_g$.
Recall the isomorphism
\bea
H^1_{dR}(X_{\bF_p}/S_{\bF_p})_{\bar 1} \to \mc V_{-h}, \qquad
\om_a^{(1)} \mapsto \bar w_a= (0,\dots, 1_b,0\dots, 0) - \frac 1n (1,\dots,1).
\eea
Under this identification, formula \eqref{cU 5.4} becomes
\bean
\label{5.5}
&&{}
\\
\notag
\Psi^{\on{KZ},-h}_k(z)(\bar w_i) 
=
\frac{-h}{C_k(z)^p}
\left(
\sum_{\ell=1}^{\left[\frac{n\tilde h}p\right]}Q_i^{(\ell p-1)}(z,h) \,z_k^{p(\ell-1)}\right)
\sum_{m=1}^{\left[\frac{n(p-\tilde h)}p\right]}Q^{(m p-1)}(z,-h) \,z_k^{p(m-1)}
\eean
where  $Q^{(m p-1)}(z,-h) = (Q^{(m p-1)}_1(z,-h), \dots, Q^{(m p-1)}_n(z,-h)) $ are $p$-hypergeometric flat sections 
of $\mc V_{-h}$.

\vsk.2>

Let $S$ be the Shapovalov form on $\mc V$. Then formula \eqref{5.5} can be written as
\bean
\label{5.6}
&&{}
\\
\notag
\Psi^{\on{KZ},-h}_k(z)(v) 
=
\frac{-h}{C_k(z)^p}
\ S\!\left( 
\sum_{\ell=1}^{\left[\frac{n\tilde h}p\right]}Q^{(\ell p-1)}(z,h) \,z_k^{p(\ell-1)}, v\right)
\sum_{m=1}^{\left[\frac{n(p-\tilde h)}p\right]}Q^{(m p-1)}(z,-h) \,z_k^{p(m-1)}
\eean
for any $v\in \mc V_{-h}$.  
Finally, we obtain the formula
\bean
\label{5.7}
&&{}
\\
\notag
\Psi^{\on{KZ},h}_k(z)(v) 
=
\frac{h}{C_k(z)^p}
\ S\!\left( 
\sum_{m=1}^{\left[\frac{n(p-\tilde h)}p\right]}Q^{(m p-1)}(z,-h) \,z_k^{p(m-1)}, v\right)
\sum_{\ell=1}^{\left[\frac{n\tilde h}p\right]}Q^{(\ell p-1)}(z,h) \,z_k^{p(\ell-1)}
\eean
for any $v\in \mc V_{h}$.

\appendix

\section{Orthogonality relations}
\label{App 1}

We give an elementary proof of a stronger version of Corollary \ref{thm ort}.

\begin{thm}
\label{thm ortA} 

The $p$-hypergeometric solutions of the equations $\kz(h)$ and $\kz(-h)$ are orthogonal
under the Shapovalov form. Namely,
for any $\ell \in \big\{1, \dots, \big[\frac{n\tilde h}p\big]\big\}$ and $
m \in \big\{1, \dots, \big[\frac{n(p-\tilde h)}p\big]\big\}$, 
we have
\bean
\label{ortA}
\big(Q^{(\ell p-1)}(z,h), Q^{(m p-1)}(z,-h)\big) = 0.
\eean

\end{thm}

\begin{proof}
By formula \eqref{do}, the left-hand side in \eqref{ortA}
lies in $\K[z^p]$.  It remains to show that the left-hand side
 cannot be a nonzero element of $\K[z^p]$\,.

Let $i,j\in\{1,\dots,n\}$. Let $Q_j^{(\ell p-1)}(z,h)$ be the $j$-th coordinate of
$Q^{(\ell p-1)}(z,h)$ and  $Q_j^{(m p-1)}(z,-h)$ the $j$-th coordinate of
$ Q^{(m p-1)}(z,-h)$. 
 Let
$Q_j(t,z,h)$  be the $j$-th coordinate of  $Q(t,z,h)$ in \eqref{Q}.
The polynomials $Q_j(t,z,h)$, $Q_j^{(\ell p-1)}(z,h)$, $Q_j^{(m p-1)}(z,-h)$ are homogeneous,
\bea
&&
\deg_{t,z} Q(t,z,h) = n\tilde h-1,
\qquad
\phantom{aaaaaaaa}
\deg_z Q_j^{(\ell p-1)}(z,h) = n\tilde h-\ell p,
\\
&&
\deg_{t,z} Q(t,z,-h) = n(p-\tilde h)-1,
\qquad
\deg_z Q_j^{(m p-1)}(z,-h) = n(p-\tilde h)-mp.
\eea
Clearly,
$\deg_{z_i} Q_j^{(\ell p-1)}(z,h)\leq \tilde h$ if $i\ne j$ and $\deg_{z_i} Q_j^{(\ell p-1)}(z,h)\leq \tilde h-1$ 
 if $i=j$, 
see \eqref{def P} and \eqref{Q}.
Similarly, 
$\deg_{z_i} Q_j^{(\ell p-1)}(z,-h)\leq p-\tilde h$ 
if $i\ne j$ and $\deg_{z_i} Q_j^{(\ell p-1)}(z,-h)\leq p-\tilde h-1$  if $i=j$.

Hence,
 if $f(z)$ is a monomial of the polynomial $Q_j^{(\ell p-1)}(z,h)$ and $g(z)$ is a monomial
   of $Q_j^{(m p-1)}(z,-h)$, then
$f(z)g(z)\in \K[z^p]$ only if 
\bea
f(z) = (z_{i_1}\dots z_{i_k})^{\tilde h} , \qquad
g(z) = (z_{i_1}\dots z_{i_k})^{p-\tilde h}
\eea
for some $1\leq i_1<\dots < i_k \leq n$. In that case we  have
\bea
k\tilde h = n\tilde h-\ell p, \qquad 
k(p-\tilde h) = n(p-\tilde h)-mp,
\eea
which can be written as 
\bean
\label{con1}
(n-k)\tilde h = \ell p, 
\\
\label{con2}
n-k = \ell+m.
\eean

Given $n, h, p, \ell$, if there is no $k$  solving \eqref{con1}, then for any $m$ equation \eqref{ortA} holds.
If there is  $k$ solving  \eqref{con1} and $m$ does not solve equation \eqref{con2} with this  $k$, then
equation \eqref{ortA} holds. If there are  $k$ and $m$ solving equations \eqref{con1} and \eqref{con2},
then the $\K[z^p]$-part of the left-hand side in \eqref{ortA} equals
\bean
\label{z su}
(n-k) \sum_{1\leq i_1<\dots < i_k \leq n} (z_{i_1}\dots z_{i_k})^p,
\eean
which is also zero in $\K[z^p]$ since $n-k$ is divisible by $p$. 
(This coefficient $n-k$ comes from counting the number of summands in the left-hand side of \eqref{ortA}
which give a particular term $(z_{i_1}\dots z_{i_k})^p$ in \eqref{z su}, see the sum in \eqref{Sha}.)
\end{proof}


\end{document}